\documentclass[journal]{IEEEtran}

\ifCLASSOPTIONcompsoc

  \usepackage[nocompress]{cite}
\else

  \usepackage{cite}
\fi

\usepackage{amsthm}
\usepackage{textcomp}
\usepackage{amsmath,amssymb,amsfonts}
\usepackage{graphicx}
\usepackage{bm}
\usepackage{multirow}
\usepackage{units}
\usepackage{url}
\usepackage{algorithm}
\usepackage{algpseudocode}
\usepackage{adjustbox}
\usepackage{makecell}  
\usepackage[most]{tcolorbox}
\usepackage{colortbl}
\usepackage{caption}
\usepackage{xcolor} 
\usepackage{balance}
\usepackage{pifont}
\usepackage{booktabs}

\definecolor{subsectioncolor}{rgb}{0, 0, 0.706}   
\definecolor{mylightblue}{rgb}{0.851, 0.898, 0.941} %
\definecolor{LightCyan}{rgb}{0.88,1,1}
\definecolor{LightOrange}{rgb}{1,0.85,0.70}

\usepackage{hyperref}
\hypersetup{
  colorlinks=true,
  citecolor=blue,
  linkcolor=red, 
  urlcolor=cyan
}

\newtheorem{Proposition}{Proposition}[section]
\newcommand{\bmx}{\bm{\mathrm{x}}}
\newcommand{\E}{\mathbb{E}}

\newcommand{\bxz}[0]{\mathbf{x}_0}
\newcommand{\bxt}[0]{\mathbf{x}_t}
\newcommand{\bxtt}[1]{\mathbf{x}_{#1}}

\renewcommand{\footnote}[1]{\oldfootnote{\scriptsize #1}}

\let\oldfootnote\footnote

\begin{document}
\title{Bi-level Guided Diffusion Models for Zero-Shot Medical Imaging Inverse Problems}

\author{Hossein Askari, Fred Roosta, Hongfu Sun
\thanks{Hossein Askari and Hongfu Sun are affiliated with the School of Electrical Engineering and Computer Science, while Fred Roosta is associated with the School of Mathematics and Physics, both at the University of Queensland, Brisbane, QLD, 4072, Australia. Emails: h.askari, fred.roosta, hongfu.sun@uq.edu.au.
}}

\maketitle

\begin{abstract}
In the realm of medical imaging, inverse problems aim to infer high-quality images from incomplete, noisy measurements, with the objective of minimizing expenses and risks to patients in clinical settings. The Diffusion Models have recently emerged as a promising approach to such practical challenges, proving particularly useful for the zero-shot inference of images from partially acquired measurements in Magnetic Resonance Imaging (MRI) and Computed Tomography (CT). A central challenge in this approach, however, is how to guide an unconditional prediction to conform to the measurement information. Existing methods rely on deficient projection or inefficient posterior score approximation guidance, which often leads to suboptimal performance. In this paper, we propose \underline{\textbf{B}}i-level \underline{G}uided \underline{D}iffusion \underline{M}odels ({BGDM}), a zero-shot imaging framework that efficiently steers the initial unconditional prediction through a \emph{bi-level} guidance strategy. Specifically, BGDM first approximates an \emph{inner-level} conditional posterior mean as an initial measurement-consistent reference point and then solves an \emph{outer-level} proximal optimization objective to reinforce the measurement consistency. Our experimental findings, using publicly available MRI and CT medical datasets, reveal that BGDM is more effective and efficient compared to the baselines, faithfully generating high-fidelity medical images and substantially reducing hallucinatory artifacts in cases of severe degradation.

\end{abstract}


\IEEEpeerreviewmaketitle

\section{Introduction}

\IEEEPARstart{C}{Contemporary} diagnostic medicine highly relies on advanced, non-invasive imaging techniques, notably Magnetic Resonance Imaging (MRI) and Computed Tomography (CT). Their unparalleled accuracy in capturing detailed anatomical measurements is of paramount importance for identifying internal abnormalities. In MRI, the Fourier transform of the spatial distribution of proton spins from the subject is acquired as measurements, which is commonly referred to as `k-space' in medical imaging contexts. In the case of CT, raw measurements, also known as `sinograms', are derived from X-ray projections obtained at various orientations around the patient. However, full k-space and sinogram acquisitions in MRI and CT often require prolonged scan durations and may pose health risks due to increased heat and radiation exposures \cite{lustig2007sparse, brenner2007computed}. In light of these implications, there have been ongoing efforts toward reducing the number of measurements, exemplified by undersampled k-spaces in MRI and sparse-view sinograms in CT. While advantageous in accelerating medical imaging procedures, sparsification and undersampling introduce difficulties in reconstructing accurate and high-quality medical images \cite{donoho2006compressed}.

Medical image reconstruction can be mathematically characterized as solving an ill-posed linear inverse problem \cite{arridge1999optical, bertero2021introduction}. The linear inverse problem is formulated as recovering an unknown target signal of interest $\textstyle{\mathbf{x} \in \mathcal{X} \subseteq  \mathbb{C}^{n}}$ from a noisy observed measurement $\textstyle{\mathbf{y} \in \mathcal{Y} \subseteq \mathbb{C}^{m}}$, given by $\mathbf{y}=\mathcal{A}\mathbf{x}+\mathbf{n}$,     
where $\textstyle{\mathcal{A} \in \mathbb{C}^{m\times n}}$ is a matrix that models a known linear measurement acquisition process (a.k.a. forward operator $\textstyle{\mathcal{A}: \mathbb{C}^{m} \rightarrow \mathbb{C}^{n})}$, and $\textstyle{\mathbf{n} \in \mathbb{C}^{m \times 1}}$ is an additive noise, simply treated here to follow the Gaussian distribution $\textstyle{\mathbf{n} \sim \mathcal{N} (\mathbf{0}, \boldsymbol{\sigma}_{\mathbf{y}}^2 \mathbf{I})}$. If the forward operator $\textstyle{\mathcal{A}}$ is singular, e.g., when $m < n$, the problem is ill-posed, indicating that the solution might not exist, be unique, or depend continuously on the measurements \cite{o1986statistical}. To mitigate the ill-posedness, it is essential to incorporate an additional assumption based on \emph{prior} knowledge to constrain the space of possible solutions. In this manner, the inverse problem then can be addressed by optimizing or sampling a function that integrates this prior or regularization term with a data consistency or likelihood term \cite{ongie2020deep}. 

A prevalent approach for prior imposition is to employ pre-trained deep generative models \cite{bora2017compressed, jalal2021robust}.
\begin{figure*}[!t]
\centering
\includegraphics[width=0.92\textwidth, height=6cm]{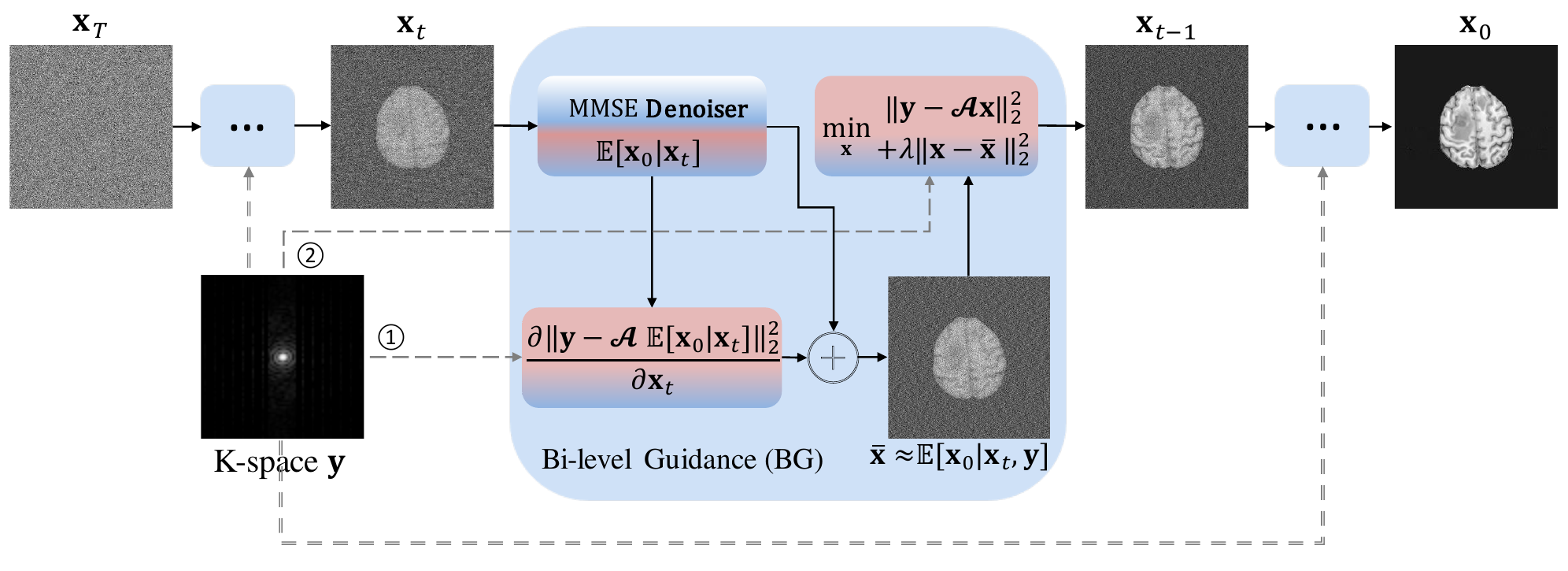}
\caption{A high-level illustration of \textbf{BGDM} that leverages pre-trained diffusion models to solve inverse problems via a \emph{bi-level} guidance (BG) strategy. BG encourages consistency at two levels by utilizing information from observed measurements $\mathbf{y}$ and the forward model $\mathcal{A}$. The \emph{inner level} \textcircled{\raisebox{-0.08pt}{1}} establishes the reference point $\bar{\mathbf{x}}$ by approximating the conditional posterior mean $\mathbb{E}[\mathbf{x}_0|\mathbf{x}_t, \mathbf{y}]$, employing gradient guidance for this purpose. Subsequently, the \emph{outer level} \textcircled{\raisebox{-0.08pt}{2}} tackles a proximal optimization objective to enforce measurement consistency further.}
\label{fig1} 
\vspace{-8pt}
\end{figure*}
A novel class of deep generative models are Diffusion Models (DMs) \cite{sohl2015deep, song2019generative, ho2020denoising, song2020score, yang2022diffusion} that have recently shown powerful capabilities in solving ill-posed inverse problems. These models are primarily designed to encode implicit prior probability distributions over data manifolds, represented as $\textstyle{\nabla_{\mathbf{x}} \log p(\mathbf{x})}$. Once trained, they can be leveraged as a chain of denoisers to produce conditional samples at inference time in a zero-shot fashion (a.k.a. \emph{plug-and-play} or \emph{training-free} approach) \cite{zhang2021plug, jalal2021robust, chung2022diffusion, wang2022zero, li2023zero}. This approach is particularly of significance in medical imaging, as measurement acquisitions can vary significantly upon such circumstances as instrumentations, scan protocols, acquisition time limit, and radiation dosage \cite{jalal2021robust, song2021solving, chung2022score}.

Top-performing methods that utilize DMs to tackle inverse problems in a zero-shot setting typically follow a three-phase progression in the iterative reverse diffusion process. Initially, they begin with an unconditional prediction, which might be either a transient noisy image \cite{song2021solving} or its denoised estimate version \cite{chung2022improving, chung2022diffusion, song2022pseudoinverse, yang2024guidance}. The subsequent phase, crucial for conditional sampling, entails guiding the initial prediction with information drawn from observed measurements. This has been accomplished via projecting images into the measurement-consistent subspaces \cite{song2021solving, lugmayr2022repaint, kawar2022denoising, wang2022zero}, approximating posterior score towards higher time-dependent likelihood \cite{chung2022diffusion, meng2022diffusion, feng2023score, fei2023generative, mardani2023variational}, and performing proximal optimization steps \cite{chung2023fast}. While the radical projection might throw the sampling trajectory off the data manifold \cite{chung2022diffusion}, and subtle score approximation may fail to generalize well to fewer timesteps \cite{song2023loss}, proximal optimization appears promising, particularly for medical imaging applications \cite{chung2023fast}. Nonetheless, the efficiency of this iterative proximal gradient-based optimization significantly diminishes in the absence of a closed-form solution \cite{chung2023fast}. Ultimately, in the third phase, the procedure progresses to the sampling, which is performed using Langevin dynamics \cite{song2020score, ho2020denoising} or more efficient samplers \cite{song2020denoising, chung2022come}.

In this paper, we introduce \textbf{B}i-level \textbf{G}uided \textbf{D}iffusion \textbf{M}odels (\textbf{BGDM}), a framework that guides the diffusion process through a \emph{bi-level} strategy. This innovative approach synergizes two distinct levels of guidance mechanisms, leveraging their unique strengths to achieve a more efficient and effective integration of measurements. To this end, we first theoretically examine the range-null space decomposition \cite{wang2022zero}, a projection-based technique, from an optimization perspective, leading us to an alternative proximal optimization objective. This \emph{outer-level} objective explicitly takes into account both data fidelity and proximity terms, where the former enforces that the reconstructed image is consistent with the acquired measurements in the transformed domains (k-space and sinograms), and the latter ensures that the solution remains close to its initial prior estimate (or reference point). Notably, this optimization problem offers a closed-form solution, which significantly increases the sampling speed compared to the previous iterative approach \cite{chung2023fast}. However, its effectiveness relies on a more accurate, and consistent reference point. To achieve this, we propose to implement an \emph{inner-level} estimate of the clean image conditioned on its noisy counterpart and the measurement. Intuitively, this \emph{bi-level} approach harmonizes the precision of the optimization objective with the robustness of the denoiser, leading to an enhanced balance between accuracy and efficiency in image reconstruction, ultimately streamlining the process while adhering to high fidelity standards (see Figure \ref{fig1}). We further propose \textbf{R}-\textbf{BGDM}, a refined version, which significantly enhances the generative quality. This is achieved by ensuring that the guided prediction does not deviate too far from the unconditional prediction on the clean manifold, thereby preserving the authenticity of the generated output. 
The contributions of our work are as follows. 
\begin{itemize}
\item \textbf{In theory}, we delve into the effective strategies tailored for addressing medical imaging inverse problems in a zero-shot setting. At the core of our approach is an assurance of data consistency achieved through analytical measures complemented by the integration of prior information extracted from pre-trained diffusion models. 
\item \textbf{In practice}, our methodology is rigorously evaluated across a spectrum of challenges, including under-sampled MRI and sparse-view CT reconstructions, as well as MRI super-resolution task. Empirical results consistently indicate that our approach surpasses the state-of-the-art (SOTA) performance of other DM-based methods, exhibiting robustness across diverse acceleration rates, projection counts, spatial resolutions, and anatomical variations (human brains, lungs, and knees). 
\end{itemize}
\section{Preliminaries}
In this section, we first review diffusion models in Subsection \ref{2.1}, and then discuss how they are used to solve inverse problems in Subsection \ref{2.2}.
\vspace{-3pt}
\subsection{Diffusion Models}\label{2.1}
A diffusion model \cite{sohl2015deep} is composed of two processes with $T$ timesteps. The first is the \textit{forward} noising process (\textit{diffusion process}), which gradually introduces Gaussian noise into the data sample $\mathbf{x}_0 \sim q(\mathbf{x}_0)$. During this procedure, a series of latent variables $\mathbf{x}_{1}, ...\mathbf{x}_{T}$ are sequentially generated, with the final one, $\mathbf{x}_T$, roughly conforming to a standard Gaussian distribution, i.e., $q(\mathbf{x}_T) \approx \mathcal{N}(\mathbf{x}_T; \mathbf{0}, \mathbf{I})$. This process is formally defined as a Markov chain
\begin{equation}\label{eq1}
\begin{aligned}
q(\mathbf{x}_{1:T} | \mathbf{x}_0) & = \prod_{t=1}^T q(\mathbf{x}_t | \mathbf{x}_{t-1}), \\
q(\mathbf{x}_t | \mathbf{x}_{t-1}) & = \mathcal{N}(\mathbf{x}_t; \sqrt{1 - \beta_t} \mathbf{x}_{t-1}, \beta_t \mathbf{I}),
\end{aligned}
\end{equation}
where $\textstyle{q(\mathbf{x}_t | \mathbf{x}_{t-1})}$ signifies the Gaussian transition kernel with a predefined variance schedule $\beta_t$. One can further compute the probabilistic distribution of $\mathbf{x}_t$  given $\mathbf{x}_0$ via reparametrization trick as $ \textstyle{q(\mathbf{x}_{t}|\mathbf{x}_{0})=\mathcal{N}(\mathbf{x}_{t}; \sqrt{\overline{\alpha}_{t}} \mathbf{x}_{0}, (1-\overline{\alpha}_{t}) \mathbf{I})}$ with $\textstyle{\alpha_{t} = 1 - \beta_{t}}$ and $\textstyle{\overline{\alpha}_{t} = \prod_{i=0}^{t} \alpha_{i}}$. Equivalently, $\mathbf{x}_t$ can be expressed as $\textstyle{\mathbf{x}_t = \sqrt{\overline{\alpha}_{t}} \mathbf{x}_0 + \boldsymbol{\sigma}_t \boldsymbol{\epsilon}}$, where $\boldsymbol{\sigma}_t = \sqrt{1-\overline{\alpha}_{t}}$  and $\textstyle{\boldsymbol{\epsilon} \sim \mathcal{N}(\mathbf{0}, \mathbf{I})}$. 
The other is the \textit{reverse} denoising process, which aims to recover the data-generating sample $\mathbf{x}_0$ by iteratively denoising the initial sample $\mathbf{x}_T$ drawn from standard Gaussian distribution $p(\mathbf{x}_T) = \mathcal{N}(\mathbf{x}_T; \mathbf{0}, \mathbf{I})$. This procedure is also  characterized by the following Markov chain: 
\begin{equation}\label{eq2}
\begin{aligned}
p_\theta(\mathbf{x}_{0:T}) &= p(\mathbf{x}_T) \prod_{t=T}^1 p_{\theta}(\mathbf{x}_{t-1} | \mathbf{x}_{t}), \\
p_\theta(\bxtt{t-1}|\bxt) &= \int_{\mathbf{x}_0} q(\bxtt{t-1} | \bxt, \bxz) p_\theta(\bxz | \bxt) d\bxz,
\end{aligned}
\end{equation}
where $\textstyle{p_{\theta}(\mathbf{x}_{t-1} | \mathbf{x}_{t})}$ is a denoising transition module with parameters $\theta$ approximating the forward posterior probability distribution $\textstyle{q(\mathbf{x}_{t-1} | \mathbf{x}_t) = q(\mathbf{x}_{t-1} | \mathbf{x}_t, \mathbf{x}_0)}$. 
The objective is to maximize the likelihood of $\textstyle{{p_{\theta}(\mathbf{x}_0)}
= {\int p_{\theta}(\mathbf{x}_{0:T}) d\mathbf{x}_{1:T}}}$. Denoising Diffusion Probabilistic Models (DDPM) \cite{ho2020denoising} assumes $\textstyle{p_{\theta}(\mathbf{x}_{t-1} | \mathbf{x}_{t}) = \mathcal{N}(\mathbf{x}_{t-1}; \bm{\mu}_{\theta}(\mathbf{x}_t, t), \bm{\sigma}_\theta(\bmx_t,t)\mathbf{I})}$ by considering $\textstyle{p_\theta(\bxz | \bxt)}$ to be a Dirac delta distribution centered at the point estimate $\textstyle{\mathbb{E}[\mathbf{x}_0|\mathbf{x}_t]}$, which is minimum mean squared error (MMSE) estimator of $\mathbf{x}_0$ given $\mathbf{x}_t$, and $\textstyle{q(\mathbf{x}_{t-1} | \mathbf{x}_t, \mathbf{x}_0)}$ to be a fixed Gaussian. Under this scheme, the training loss $\textstyle{\ell({\theta})}$ an be simply formulated as
\begin{equation}\label{eq3}
\min_{\theta} \ell(\theta) := \min_{\theta} \mathbb{E}_{t \sim (\mathbf{0}, {T}), \mathbf{x}_0 \sim q(\mathbf{x}_0), \bm{\epsilon} \sim \mathcal{N}(\mathbf{0}, \mathbf{I})} \big[ \Vert \bm{\epsilon} - \bm{\epsilon}_{\theta} (\mathbf{x}_t, t)\Vert_2^2\big].
\end{equation}
Therefore, given the trained denoising function $\textstyle{\bm{\epsilon}_{\theta}(\mathbf{x}_t, t)}$, samples can be generated using DDPM, Denoising Diffusion Implicit Models (DDIM) \cite{song2020denoising}, or other solvers \cite{lu2022dpm, zhang2022fast}. 
\vspace{-10pt}
\subsection{Diffusion Models for Linear Inverse Problems}\label{2.2}
An inverse problem seeks to estimate an unknown image $\mathbf{x}$ from partially observed, noisy measurement $\mathbf{y}$. They are generally approached by optimizing or sampling a function that combines a term for data fidelity or likelihood with a term for regularization or prior \cite{ongie2020deep}. A detailed exploration of methods for solving linear inverse problems can be found in Appendix \ref{A1}. A common method for regularization involves using pre-trained priors from generative models. Recently, pre-trained diffusion models \cite{ho2020denoising, nichol2021improved} have been leveraged as a powerful generative prior (a.k.a. denoiser), in a zero-shot fashion, to efficiently sample from the conditional posterior. Due to their unique characteristics, namely the ability to model complex distributions, the efficient iterative nature of the denoising process, and the capacity to effectively conduct conditional sampling, these models stand out as a potent solution for solving inverse problems \cite{daras2022score, rombach2022high}. A primary difficulty, however, is how to guide the unconditional prediction to conform to the measurement information in each iteration. Methods addressing this generally fall into two distinct categories as follows.
\vspace{4pt}

\textbf{Posterior Score Approximation.}\label{psa} The reverse Stochastic Differential Equation (SDE) for a conditional generation can be written as  
\begin{equation}\label{eq4}
\mathrm{d}\mathbf{x}_{t}=\big[\mathbf{f}(\mathbf{x}_{t},t)-g^{2}(t){\nabla_{\mathbf{x}_t}\log{p(\mathbf{x}_{t}|\mathbf{y})}}\big]\mathrm{d}\bar{t}+g(t)\mathrm{d}\mathbf{\bar{w}}_{t},
\end{equation}
where ${\nabla_{\mathbf{x}_t}\log{p(\mathbf{x}_{t}|\mathbf{y})}}$ is referred to as a posterior score that can be decomposed through Bayesian' rule as follows. 
\begin{equation}\label{eq5}
\nabla_{\mathbf{x}_t} \log p(\mathbf{x}_t|\mathbf{y}) = \nabla_{\mathbf{x}_t} \log p(\mathbf{x}_t) + \nabla_{\mathbf{x}_t} \log p(\mathbf{y}|\mathbf{x}_t).
\end{equation}
The composite score results from the prior score combined with the time-dependent likelihood score. While one can closely approximate the prior score with a pre-trained diffusion model, i.e., {{$\textstyle{\nabla_{\mathbf{x}_t} \log p(\mathbf{x}_t) \simeq \frac{-1}{\sqrt{1-\overline{\alpha}_t}}\boldsymbol{\epsilon}_{\mathbf{\theta}}(\mathbf{x}_t, t)}$}}, the likelihood score is analytically intractable to compute. This becomes evident when considering $\textstyle{p(\mathbf{y}|\mathbf{x}_t) = \int_{\mathbf{x}_0} p(\mathbf{y}|\mathbf{x}_0) p(\mathbf{x}_0|\mathbf{x}_t) \, d\mathbf{x}_0}$ according to the graphical inferences $\textstyle{\mathbf{x}_0 \rightarrow \mathbf{y}}$ and $\textstyle{\mathbf{x}_0 \rightarrow \mathbf{x}_t}$. The measurement models can be represented by $p(\mathbf{y}|\mathbf{x}_0)=\mathcal{N}(\mathcal{A}\mathbf{x}_0, \boldsymbol{\sigma}_{\mathbf{y}}^2)$. The intractability of $\textstyle{p(\mathbf{y}|\mathbf{x}_t)}$ arises from $\textstyle{p(\mathbf{x}_0|\mathbf{x}_t)}$. Several strategies have been proposed to approximate the likelihood term. Among the most prevalent are DPS \cite{chung2022diffusion} and \( \Pi \)GDM \cite{song2022pseudoinverse}, where point estimate \( \textstyle{p(\mathbf{x}_0 | \mathbf{x}_t) = \delta (\mathbf{x}_0 - \mathbf{x}_{0|t})} \) and Gaussian assumption \( \textstyle{p(\mathbf{x}_0 | \mathbf{x}_t) \sim \mathcal{N} (\mathbf{x}_{0|t}, \nicefrac{\boldsymbol{\sigma}_t^2}{\boldsymbol{\sigma}_t^2 +1} \ \mathbf{I})} \) are considered respectively to estimate $\textstyle{p(\mathbf{y} | \mathbf{x}_t)}$. The term $\textstyle{\mathbf{x}_{0|t}}$ is posterior mean (or denoised estimate) of $\textstyle{\mathbf{x}_{0}}$ conditioned on $\textstyle{\mathbf{x}_{t}}$, defined as $\textstyle{\mathbf{x}_{0|t}:= \mathbb{E}[\mathbf{x}_0|\mathbf{x}_t]} = \mathbb{E}_{\mathbf{x}_0 \sim p(\mathbf{x}_0|\mathbf{x}_t)}[\mathbf{x}_0]$. As a result, the likelihood score can be reformulated as
\vspace{-2pt}
\begin{align}\label{eq6}
\nabla_{\mathbf{x}_t} \log p(\mathbf{y}|\mathbf{x}_t) \simeq  \underbrace{\frac{\partial \ (\mathbf{x}_{0|t})}{\partial \mathbf{x}_t}}_{{\rm J}}\underbrace{\mathcal{H}(\mathbf{y} - \mathcal{A} \ {\mathbf{x}_{0|t}})}_{{\rm V}},
\end{align}
\vspace{-2pt}
which is essentially a Vector ($\rm V$)-Jacobian ($\rm J$) Product that enforces consistency between the denoising result and the measurements, with \(\textstyle{\mathcal{H}}\) corresponding to \(\textstyle{\mathcal{A}^{\top}}\) in DPS and to \(\textstyle{\mathcal{A}^{\dagger}}\) (the Moore-Penrose pseudoinverse of $\textstyle{\mathcal{A}}$) in $\Pi$GDM. These methods suffer from inconsistent reliability, sometimes delivering excellent results and at other times failing to create realistic images. They are also notably slow due to their dependence on a high number of diffusion steps for satisfactory outcomes. In the context of MRI reconstruction in medical imaging, DPS leads to noisy outputs \cite{chung2023fast}. More recently, variational posterior approximation has been proposed \cite{mardani2023variational}, yet it requires computationally expensive test-time optimization.
\vspace{4pt}

\textbf{Decomposition / Projection.} 
Denoising Diffusion Restoration Model (DDRM) \cite{kawar2022denoising} attempted to solve inverse problems in a zero-shot way using singular value decomposition (SVD) of $\mathcal{A}$. However, for medical imaging applications with complex measurement operators, the SVD decomposition can be prohibitive \cite{chung2023fast}. In ScoreMed \cite{song2021solving}, an alternative decomposition method for $\mathcal{A}$ within the sampling process is proposed, suitable for medical imaging, assuming that $\mathcal{A}$ is of full rank. Denoising Diffusion Null-Space Models (DDNM) \cite{wang2022zero} introduces a range-null space decomposition for zero-shot image reconstruction, where the range space ensures data consistency, and the null space enhances realism. Both ScoreMed and DDNM essentially use back-projection tricks \cite{tirer2020back} to meet the measurement consistency in a non-noisy measurement scenario, which can be expressed as:
\vspace{-2pt}
\begin{equation}\label{eq7}
\hat{\mathbf{x}}_{t} =\sqrt{\overline{\alpha}_{t}} \big( {{\mathcal{A}^{\dagger}\mathbf{y}}+{(\mathbf{I}-\mathcal{A}^{\dagger}\mathcal{A})}{\mathbf{x}_{0|t}}}\big)+ \boldsymbol{\sigma}_t\boldsymbol{\epsilon},
\end{equation}
where the extra noise $\boldsymbol{\sigma}_t\boldsymbol{\epsilon}$ is excluded in DDNM, yielding a higher performance. However, these projection-based methods frequently encounter challenges in maintaining the sample's realness, as the projection might shift the sample path away from the data manifold \cite{chung2022improving}.

\section{Proposed Method}
\subsection{Bi-level Guided Diffusion Model (BGDM)}
We motivate our approach by highlighting two critical drawbacks inherent in projection-based methods, especially in DDNM, which utilizes the range-null space decomposition to construct a general solution $\hat{\mathbf{x}}$ as 
\begin{equation}\label{eq8}
\hat{\mathbf{x}} = {\mathcal{A}^{\dagger}\mathbf{y}}+{(\mathbf{I}-\mathcal{A}^{\dagger}\mathcal{A})}{\bar{\mathbf{x}}},
\end{equation}
where ${\bar{\mathbf{x}}}$ can be chosen arbitrarily from $\mathbb{C}^{n}$ without affecting the consistency. The foundational interplay between these spaces is evident: the range space, represented by \( \mathcal{A}^{\dagger}\mathbf{y} \), embodies the solution components originating from observations, whereas the null space, denoted by \( (\mathbf{I}- \mathcal{A}^{\dagger}\mathcal{A})\bar{\mathbf{x}} \), encompasses the solution's unobserved elements. We illuminate a new interpretation of this decomposition from an optimization perspective in the following proposition.
\begin{Proposition}\label{prop:GD}
Consider the least squares problem $\min_{\mathbf{x} \in \mathbb{R}^{n}} \Vert \mathbf{y} -\mathcal{A}\mathbf{x}  \Vert_2^2$ where $\mathcal{A} \in \mathbb{R}^{m \times n}$ is any matrix and $\mathbf{y} \in \mathbb{R}^m$. Gradient descent, initialized at ${\bar{\mathbf{x}}} \in \mathbb{R}^{n}$ and with small enough learning rate, converges to $\hat{\mathbf{x}} = {\mathcal{A}^{\dagger}\mathbf{y}}+{(\mathbf{I}-\mathcal{A}^{\dagger}\mathcal{A})\bar{\mathbf{x}}}$.
\end{Proposition}
\begin{proof}
Consider an iteration of gradient descent, initialized from $\mathbf{x}^{(0)}$, on the least squares problem
\begin{align*}
\mathbf{x}^{(t+1)} &= \mathbf{x}^{(t)} + \alpha \mathcal{A}^T(\mathbf{y} -\mathcal{A}\mathbf{x}^{(t)}).
\end{align*}
Defining $\mathbf{r}^{(t)} = \mathbf{y} -\mathcal{A}\mathbf{x}^{(t)}$, it follows that
\begin{align*}
\mathbf{r}^{(t+1)} &= \left(\mathbf{I} - \alpha \mathcal{A} \mathcal{A}^T \right)\mathbf{r}^{(t)} = \ldots = \left(\mathbf{I} - \alpha \mathcal{A} \mathcal{A}^T \right)^{t+1}\mathbf{r}^{(0)}.
\end{align*}
Hence,
\begin{align*}
\mathbf{x}^{(t+1)} &= \mathbf{x}^{(t)} + \alpha \mathcal{A}^T\left(\mathbf{I} - \alpha \mathcal{A} \mathcal{A}^T \right)^{t}\mathbf{r}^{(0)} \\
&= \mathbf{x}^{(0)} + \alpha \mathcal{A}^T \sum_{i=0}^{t} \left(\mathbf{I} - \alpha \mathcal{A} \mathcal{A}^T \right)^{i}\mathbf{r}^{(0)} \\
&= \mathbf{x}^{(0)} + \alpha \sum_{i=0}^{t} \left(\mathbf{I} - \alpha \mathcal{A}^T \mathcal{A} \right)^{i} \mathcal{A}^T \mathbf{r}^{(0)}.
\end{align*}
Subsequently, as long as $0 < \alpha < 1/\|\mathcal{A}\|^{2}$, from \cite[Theorem 16]{ben1963contributions}, we get
\begin{equation}
\begin{aligned}
\lim_{t \to \infty} \mathbf{x}^{(t)} &= \mathbf{x}^{(0)} + \alpha \sum_{i=0}^{\infty} \left(\mathbf{I} - \alpha \mathcal{A}^T \mathcal{A} \right)^{i} \mathcal{A}^T \mathbf{r}^{(0)} \\
&= \mathbf{x}^{(0)}  + \mathcal{A}^{\dagger} \mathbf{r}^{(0)}.
\end{aligned}
\end{equation}
This concludes the proof. 
\end{proof}
Proposition \ref{prop:GD} highlights the behavior of gradient descent on a least squares problem when initiated from any initial point, in particular $\bar{\mathbf{x}} = \mathbf{x}_{0|t}$. The solution, upon convergence, can be expressed as
\begin{equation}\label{eq9}
\hat{\mathbf{x}}_{0|t,\mathbf{y}}:=\hat{\mathbf{x}} = \mathbf{x}_{0|t} + \mathcal{A}^\dagger(\mathbf{y}-\mathcal{A}\mathbf{x}_{0|t}).
\end{equation}
Here, the term $\mathcal{A}^\dagger(\mathbf{y}-\mathcal{A}\mathbf{x}_{0|t})$ represents the correction applied to the initial prediction, factoring in the difference between predicted and observed measurements. However, this method is not devoid of challenges. Primarily, the correction term, solely determined by $(\mathbf{y}-\mathcal{A}\mathbf{x}_{0|t})$, can be significantly affected if $\mathbf{y}$ is noisy, potentially leading our estimates astray. Furthermore, considering the proposition outlined, relying exclusively on gradient-based correction for adjusting the initial estimate, particularly when  $\mathbf{x}_{0|t}$ hold uncertainties, can predispose the optimization towards suboptimal solutions. These points highlight the need for a more reliable optimization path. Hence, we define the decomposition \eqref{eq9} explicitly by embedding a regularization term into our optimization objective, acting as a penalty against large deviations from our initial estimate. This results in the following \emph{outer-level} proximal optimization objective:
\begin{equation}\label{eq10}
\hat{\mathbf{x}}_{0|t,\mathbf{y}} = \underset{\mathbf{x}}{\operatorname{arg\, min}} \ \frac{1}{2} \underbrace{\Vert\mathbf{y -\mathcal{A}\mathbf{x}}  \Vert_{2}^{2}}_{\text{term (i)}} +\frac{{\color{orange}\lambda}}{2}  \underbrace{\Vert\mathbf{x}-{\mathbf{x}_{0|t}}\Vert_{2}^{2}}_{\text{term (ii)}},
\end{equation}
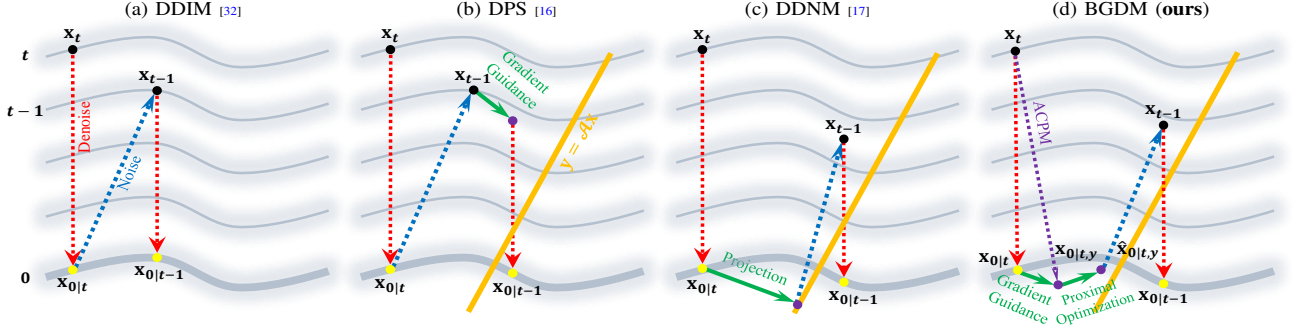
\begin{figure*}[!t]
\centering
\begin{adjustbox}{width=0.96\linewidth, height=4.5cm}
\begin{tikzpicture}
\draw (0, 0) node[inner sep=0] {\includegraphics[width=0.96\textwidth]{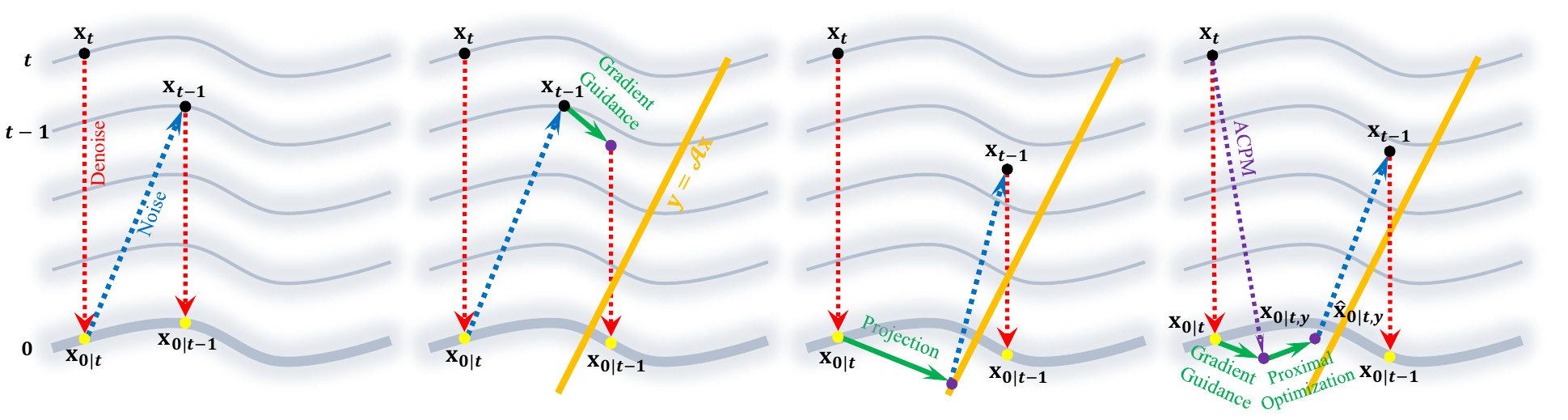}};
\draw (-6.3, 2.3) node {\fontsize{8pt}{9pt}\selectfont (a) DDIM \fontsize{5pt}{3pt}\selectfont \cite{song2020denoising}};
\draw (-2, 2.3) node {\fontsize{8pt}{9pt}\selectfont (b) DPS \fontsize{5pt}{3pt}\selectfont \cite{chung2022diffusion}};
\draw (2, 2.3) node {\fontsize{8pt}{9pt}\selectfont (c) DDNM \fontsize{5pt}{3pt}\selectfont \cite{wang2022zero}};
\draw (6.3, 2.3) node {\fontsize{8pt}{9pt}\selectfont (d) BGDM (\textbf{ours})};
\end{tikzpicture}
\end{adjustbox}
\caption{An illustration of the geometric principles underpinning diffusion samplers and various guidance schemes. (a) DDIM is an unconditional diffusion sampler devoid of guidance. (b) DPS employs gradient guidance with possible deviation from the accurate manifold. (c) DDNM projects denoised samples into a measurement-consistent subspace. (d) our proposed method employs a \emph{bi-level} guidance strategy; the \emph{inner level} approximates the initial prediction with a conditional posterior mean through gradient guidance, while the \emph{outer level} tackles an optimization problem to impose measurement consistency further. Note that ACPM stands for Approximated Conditional Posterior Mean derived in \eqref{eqnew2}.}
\label{fig2}
\vspace{-14pt}
\end{figure*}
where the fidelity term (i) aims to minimize the discrepancy between the predicted and observed measurements, while the proximity term (ii) penalizes deviations from the initial estimate. The regularization parameter ${\color{orange}\lambda}$ offers a balance between fidelity and proximity, ensuring our new estimate aligns with observations while respecting our initial belief encapsulated in $\mathbf{x}_{0|t}$. Note that the solution $\hat{\mathbf{x}}_{0|t,\mathbf{y}}$ can be regarded as a Maximum a Posteriori (MAP) estimate of a Bayesian posterior with a Gaussian prior centered around $\mathbf{x}_{0|t}$. Typically, this solution is available in a closed form. For MRI reconstruction, the details can be found in appendix \ref{A3}.
\begin{figure}[!t]
\begin{algorithm}[H]
\scriptsize
\caption{BGDM Sampling}\label{alg1}
\begin{algorithmic}[1] 
\Require {The measurement $\mathbf{y}$, the forward operator $\mathcal{A}$, and pre-trained denoiser $\boldsymbol{\epsilon}_{\mathbf{\theta}}(\cdot,\cdot)$}
\State $\mathbf{x}_T \sim \mathcal{N} (0, \mathbf{I})$ 
\For {$t=T, \ldots , 1$}
\State $\overline{\alpha}_{t-1} \gets {1-\boldsymbol{\sigma}_t^2}$ \Comment{Get $\overline{\alpha}_{t-1}$ in VPSDE}
\State $c_1 \gets {\eta} \sqrt{1-\overline{\alpha}_{t-1}}$ 
\Comment{Get coefficients $c_1$ and $c_2$ in DDIM}
\State $c_2 \gets  {\sqrt{1-{\overline{\alpha}_{t-1}}-c_1^2}}$
\State $\boldsymbol{\epsilon} \sim \mathcal{N} (0, \mathbf{I})$ \textbf{if} $t > 0$, \textbf{else} $\boldsymbol{\epsilon}=0$\Comment{Sample \emph{i.i.d.} Gaussian}
\State $\mathbf{x}_{0|t} \gets \frac{1}{\sqrt{\overline{\alpha}}_{t}}\left( \mathbf{x}_{t} - \sqrt{1-\overline{\alpha}_{t}} \boldsymbol{\epsilon}_{\mathbf{\theta}}(\mathbf{x}_{t},t) \right)$\Comment{Predict one-step MMSE denoiser result, \eqref{eq11}}
\State ${\mathbf{x}}_{0|t,\mathbf{y}} \gets \mathbf{x}_{0|t} - {\color{orange}\zeta} \nabla_{\mathbf{x}_t} \Vert\mathbf{y} -\mathcal{A}\mathbf{x}_{0|t}\Vert_2^2$\label{step8}
\Comment{The \emph{inner-level} guidance \eqref{eqnew2}}
\State $\hat{\mathbf{x}}_{0|t,\mathbf{y}} \gets \underset{\mathbf{x}}{\operatorname{arg\, min}} \big\{ {\frac{1}{2}\Vert\mathbf{y}-\mathcal{A}\mathbf{x}\Vert_{2}^{2}+\frac{{\color{orange}\lambda}}{2} \Vert\mathbf{x}-{\mathbf{x}}_{0|t,\mathbf{y}} \Vert_{2}^{2}}\big\}$ \label{step9}
\Comment{\emph{Outer-level} objective \eqref{eqnew}}
\State $\mathbf{x}_{t-1} \gets {\sqrt{\overline{\alpha}_{t-1}}}\hat{\mathbf{x}}_{0|t,\mathbf{y}} + (c_1 \boldsymbol{\epsilon} +c_2 \boldsymbol{\epsilon}_{\mathbf{\theta}}(\mathbf{x}_t, t))$
\Comment{BGDM update based on DDIM sampling}
\EndFor
\State \textbf{return} $\mathbf{x}_{0}$
\end{algorithmic} 
\end{algorithm}
\vspace{-14pt}
\end{figure}
Secondly, as previously noted, different choices of $\bar{\mathbf{x}}$ in \eqref{eq8} result in estimates that are all equally consistent, and the choice of $\mathbf{x}_{0|t}$ represents just one specific solution among the possibilities. We postulate that the chosen $\bar{\mathbf{x}}$ can profoundly influence the trajectory of the projections. By strategically choosing $\bar{\mathbf{x}}$ in accordance with the actual measurement $\mathbf{y}$, we can make our solutions more efficient and accurate, yet ensuring that they respect the desired data distribution $q(\mathbf{x})$. In a similar reasoning, the effectiveness of the {proximity} {term} in \eqref{eq10} highly relies on the quality of the {prior} $\mathbf{x}_{0|t}$. If the {prior} is not a desirable {estimate}, it might mislead the {optimization}. To identify a solution, we return to the posterior mean of $\mathbf{x}_0$ given $\mathbf{x}_t$ discussed in Section \ref{psa}. For Variance Preserving  SDE (VPSDEs), the posterior mean is driven based on Tweedie's formula as
\begin{equation}\label{eq11}
{\mathbf{x}_{0|t}} =\mathbb{E}[{{\mathbf{x}_0|\mathbf{x}_t}}]  =\frac{1}{{\sqrt{\overline{\alpha}_{t}}}}\big(\mathbf{x}_{t}+({1-\overline{\alpha}_{t}}){\nabla_{{{\mathbf{x}_t}}}\log p({\mathbf{x}_t})}\big),
\end{equation}
which is further extended in \cite{ravula2023optimizing} with an additional measurement $\mathbf{y}$ for Variance Exploding SDE (VESDEs). The updated formula for the conditional posterior mean in VPSDEs (see Appendix \ref{A41} ), can be presented as 
\begin{equation}\label{eq12}
{\mathbf{x}}_{0|t,\mathbf{y}}=\mathbb{E}[{\mathbf{x}_0|\mathbf{x}_t, \mathbf{y}}] = \mathbb{E}[{{\mathbf{x}_0|\mathbf{x}_t}}] + \frac{1-\overline{\alpha}_{t}}{{\sqrt{\overline{\alpha}_{t}}}} {\nabla_{{{\mathbf{x}_t}}}\log p({\mathbf{y}}|\mathbf{x}_t)}.
\end{equation}
This conditional estimation for the initial unconditional prediction functions as an \emph{inner-level} guidance. Now, it becomes clear that by integrating the prior score with the likelihood score, we can procure a consistent estimate of reference points than by solely relying on the prior score. Note that the \emph{inner-level} guidance can also be interpreted as an optimization procedure via gradient descent aimed at minimizing the guidance loss $\mathcal{L}_{t} = -\log p(\mathbf{y}|\mathbf{x}_t)$ in the vicinity of the denoised sample $\mathbf{x}_{t}$. Due to access to only $\mathcal{L}_0$, we follow DPS \cite{chung2022diffusion} which uses a clean data estimation $\mathbf{x}_{0|t}$ as a point estimate of the true loss term, i.e., $\textstyle{\mathcal{L}_{t} =- \log p(\mathbf{y}|\mathbf{x}_{0|t})}$. If the measurement noise is Gaussian, i.e., \(\textstyle{\mathbf{y} \sim \mathcal{N}(\mathbf{y}; \mathcal{A}(\mathbf{x}_0), \boldsymbol{\sigma}^2_{\mathbf{y}} \mathbf{I})}\), we then have $\textstyle{ \nabla_{\mathbf{x}_t} \mathcal{L}_{t} \approx  \frac{1}{\boldsymbol{\sigma}^2_{\mathbf{y}}} \nabla_{\mathbf{x}_t} \|\mathbf{y}-\mathcal{A}(\mathbf{x}_{0|t}) \|_2^2}$. In practice, it is assumed that $\textstyle{p(\mathbf{y}|\mathbf{x}_{0|t}) \sim \mathcal{N} (\mathbf{y}; \mathcal{A}\mathbf{x}_{0|t}, \boldsymbol{\sigma}_t^2\mathbf{I})}$.
Based on DPS's result, an approximation of the expectation in \eqref{eq12} can be established (see Appendix \ref{A42}) as
\vspace{-4pt}
\begin{equation}\label{eqnew2}
{\mathbf{x}}_{0|t,\mathbf{y}} \simeq \frac{1}{{\sqrt{\overline{\alpha}_{t}}}}  \big[\mathbf{x}_t - \sqrt{1-\overline{\alpha}_{t}} \boldsymbol{\epsilon}_{\mathbf{\theta}}(\mathbf{x}_t, t) - 
{\color{orange}\zeta} \nabla_{\mathbf{x}_t}\Vert \mathbf{y} - \mathcal{A}\mathbf{x}_{0|t}\Vert_2^2
\big],
\end{equation}
where ${\color{orange}\zeta}$ is a likelihood step size. Therefore, \eqref{eq10} will be redefined as 
\begin{equation}\label{eqnew}
\hat{\mathbf{x}}_{0|t,\mathbf{y}} = \underset{\mathbf{x}}{\operatorname{arg\, min}} \ \frac{1}{2} {\Vert\mathbf{y -\mathcal{A}\mathbf{x}}  \Vert_{2}^{2}}+\frac{{\color{orange}\lambda}}{2}  {\Vert\mathbf{x}-{\mathbf{x}_{0|t, \mathbf{y}}}\Vert_{2}^{2}}.
\end{equation}
Combining \eqref{eqnew2} and \eqref{eqnew}, we call our method, \textbf{B}i-level \textbf{G}uided \textbf{D}iffusion \textbf{M}odels (\textbf{BGDM}). 
For sampling $\mathbf{x}_{t-1}$, we employ DDIM \cite{song2020denoising}, an accelerated diffusion sampling method, which transitions the stochastic ancestral sampling of DDPM to deterministic sampling, thereby expediting the sampling process.

From the discussion presented above, we summarized the steps of our proposed method in Algorithm \ref{alg1}. We also provided a schematic illustration of the geometrical differences between our BGDM guidance strategy and other SOTA guidance techniques in Figure \ref{fig2}.
\vspace{-6pt}
\begin{figure*}[!t]
\centering    
\includegraphics[width=1\linewidth, height=0.16\textheight]{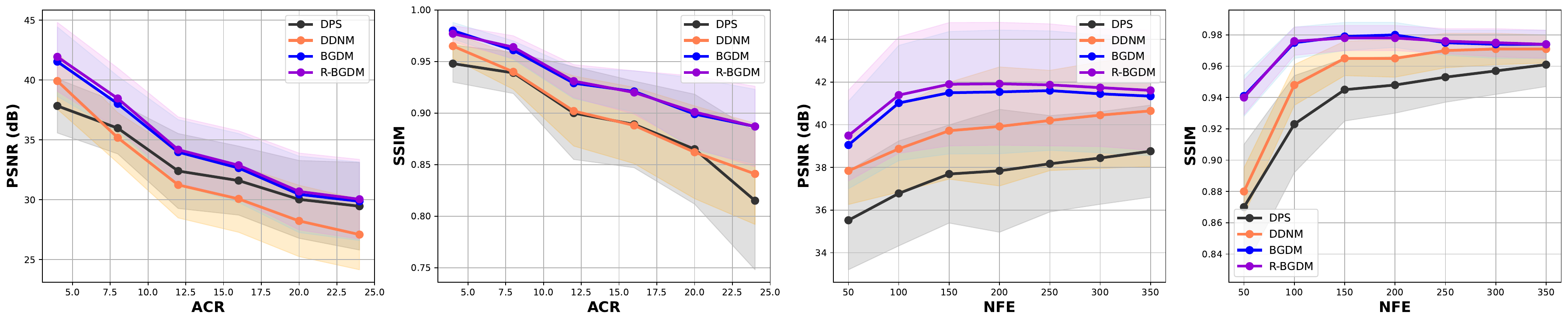}
\captionsetup{font=small, aboveskip=0pt, belowskip=0pt}
\caption{In the horizontal array of graphs from left to right, the first two graphs illustrate the BraTs undersampled MRI reconstruction results for 200 timesteps at various acceleration rates, and the last two graphs display the results over a span of 350 timesteps at a fixed ACR of 4.}
\label{fig3}
\vspace{-4pt}
\end{figure*}
\begin{figure*}[ht]
\begin{adjustbox}{width=1\linewidth}
\begin{tikzpicture}
\draw (0, 0) node[inner sep=0] {\includegraphics[width=1\textwidth]{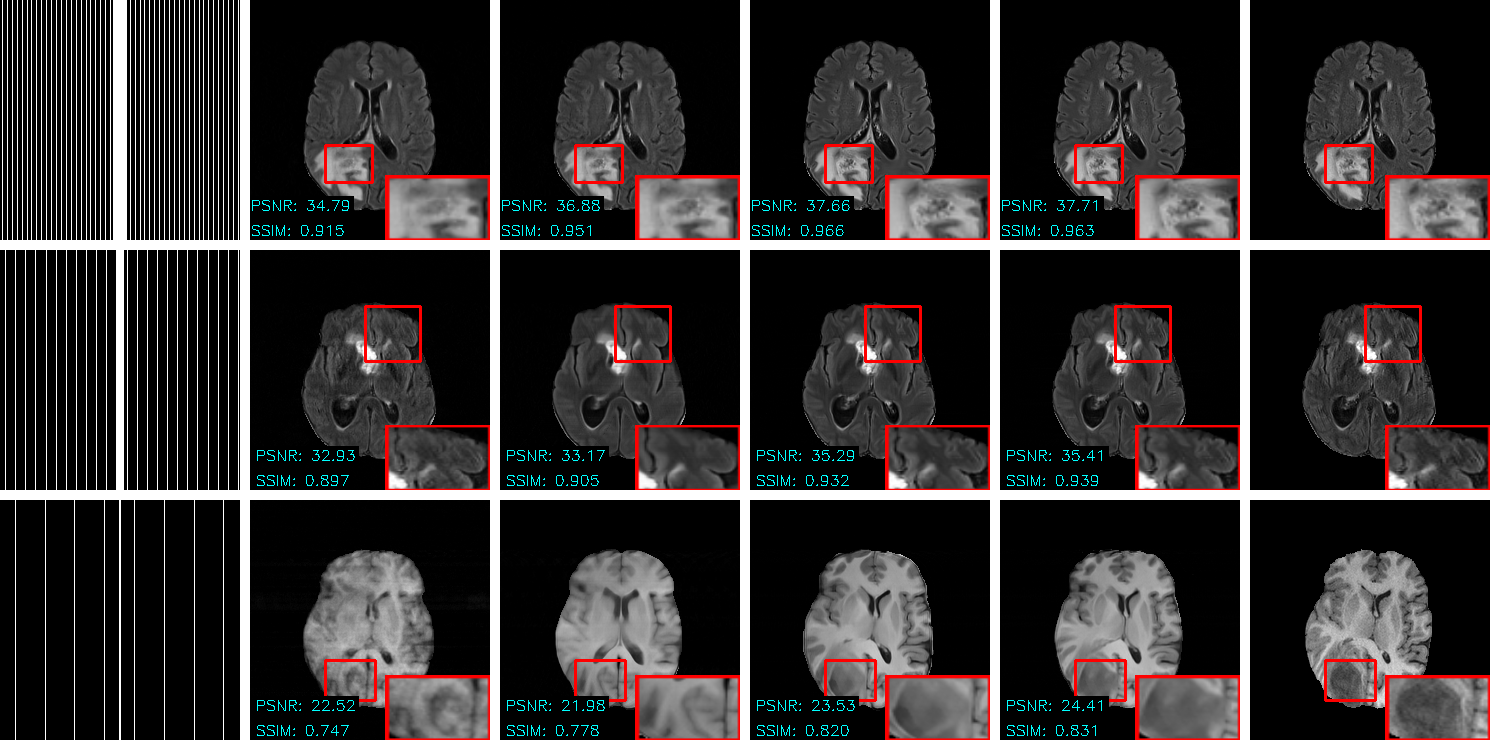}};
\draw (-7.5, 4.7) node {\fontsize{8pt}{9pt}\selectfont  Masks};
\draw (-4.5, 4.7) node {\fontsize{8pt}{9pt}\selectfont DPS}; 
\draw (-1.5, 4.7) node {\fontsize{8pt}{9pt}\selectfont DDNM};
\draw (1.5, 4.7) node {\fontsize{8pt}{9pt}\selectfont {BGDM} (\textbf{ours})};
\draw (4.5, 4.7) node {\fontsize{8pt}{9pt}\selectfont {R-BGDM} (\textbf{ours})};
\draw (7.5, 4.7) node {\fontsize{8pt}{9pt}\selectfont References};
\draw (-9.3, 3) node[rotate=90] {\fontsize{8pt}{9pt}\selectfont ACR= 4};
\draw (-9.3, 0) node[rotate=90] {\fontsize{8pt}{9pt}\selectfont ACR= 8};
\draw (-9.3, -3) node[rotate=90] {\fontsize{8pt}{9pt}\selectfont ACR= 24};
\end{tikzpicture}
\end{adjustbox}
\caption{The qualitative results of undersampled MRI reconstruction on the BraTS dataset, depicted for ACR 4, 8, and 24.}
\label{fig4}
\vspace{-14pt}
\end{figure*}
\subsection{Refined Version: R-BGDM}
While \emph{bi-level} guidance strategy provides enhanced efficiency and performance compared to standard baselines, it remains susceptible to manifold deviation issues due to its dependence on gradient guidance, generally speaking, loss guidance approximation \cite{he2023manifold, bansal2023universal, yu2023freedom}. To mitigate this issue, which has been empirically \cite{he2023manifold} and theoretically  \cite{yang2024guidance} verified, Manifold Preserving Guided Diffusion \cite{he2023manifold}, along with other works \cite{song2023solving, rout2024solving}, addresses the inverse problem in the latent space. However, this approach requires the training of an additional auto-encoder. Alternatively, we propose to further impose constraints on the guided prediction by incorporating a `refinement gradient' term. This term is applied at the \emph{outer-level} of guidance, ensuring that the guided predictions remain close to the initial, unconditioned ones, which is on the clean manifold. This strategic balance appears to play a significant role in maintaining the samples' authenticity and consistency. To formally introduce this concept, we consider the term as  $\Vert\hat{\mathbf{x}}_{0|t,\mathbf{y}} - {{\mathbf{x}}_{0|t}}\Vert_2^2$, which aims to measure the discrepancy between the \emph{outer-level} outputs and the original predictions. This discrepancy can be used as an additional piece of information to boost the model's ability to generate realistic images. To seamlessly integrate this insight into the BGDM framework, we can further update the \emph{outer-level} objective (step \ref{step9} of Alg \ref{alg1}) using the following formulation:
\begin{equation}\label{eq1515}\hat{\mathbf{x}}_{0|t,\mathbf{y}} = \hat{\mathbf{x}}_{0|t,\mathbf{y}} -{\color{orange}\gamma} \nabla_{\mathbf{x}_{0|t}} \Vert \hat{\mathbf{x}}_{0|t,\mathbf{y}} - {{\mathbf{x}}_{0|t}}\Vert_2^2, \end{equation} where ${\color{orange}\gamma}$ denotes the intensity of the `refinement gradient' adjustment. This modification is pivotal in refining the BGDM framework to yield more authentic and aligned samples. It should be noted that in \eqref{eq1515}, the gradient is computed with respect to ${\mathbf{x}}_{0|t}$ on clean manifold, rather than ${\mathbf{x}}_{t}$.
\section{Experiments}
In this section, we initially present the experimental setup and the implementation details of our model. Subsequently, we provide the results, where we quantitatively and qualitatively compare our model with SOTA methods. The ablation study is discussed in the final subsection. 
\vspace{-4pt}
\begin{table*}[t!]
\small
\centering
\caption{Quantitative results for the fastMRI knee dataset, examining performance across various mask types and acceleration rates.}
\begin{adjustbox}{max width=\textwidth}
\begin{tabular}{c|cc|cc|cc|cc}
\hline
\multirow{2}{*}{Method} & \multicolumn{4}{c|}{\textbf{Uniform1D}} & \multicolumn{4}{c}{\textbf{Gaussian1D}} \\ \cline{2-9} 
 & \multicolumn{2}{c|}{4$\times$ ACR} & \multicolumn{2}{c|}{8$\times$ ACR} & \multicolumn{2}{c|}{4$\times$ ACR} & \multicolumn{2}{c}{8$\times$ ACR} \\ \cline{2-9} 
 & PSNR$\uparrow$         & SSIM$\uparrow$         & PSNR$\uparrow$         & SSIM$\uparrow$         & PSNR$\uparrow$         & SSIM$\uparrow$         & PSNR$\uparrow$         & SSIM$\uparrow$        \\ \hline
DPS \cite{chung2022diffusion} & 32.40$\pm$2.19            & 0.843$\pm$0.063        &31.07$\pm$2.32              &  0.804$\pm$ 0.073             &34.93$\pm$1.90             &0.882$\pm$0.063             &33.72$\pm$1.97              &0.853$\pm$0.071            \\
DDNM \cite{wang2022zero}& 33.66$\pm$2.59             &0.857$\pm$0.048            &32.01$\pm$2.84             &0.821$\pm$0.062           &37.51$\pm$2.33              &0.899$\pm$0.046             &35.71$\pm$2.42              &0.871$\pm$0.054             \\
Score-MRI & 31.95$\pm$1.45              &0.812$\pm$0.036              &27.97$\pm$2.03              &0.738$\pm$0.053              &33.96$\pm$1.27              &0.858$\pm$0.028              &30.82$\pm$1.37              &0.762$\pm$0.034              \\
\hline 
\rowcolor{LightOrange}
BGDM (\textbf{ours})           &34.38$\pm$2.66              &0.868$\pm$0.051              &32.59$\pm$2.91                 &0.831$\pm$0.063              &37.92$\pm$2.47              &0.904$\pm$0.049             &36.30$\pm$2.54              &0.881$\pm$0.056             \\
\rowcolor{mylightblue}
R-BGDM (\textbf{ours})         &\textbf{34.73}$\pm$\textbf{2.73}              & \textbf{0.875}$\pm$\textbf{0.050}              &\textbf{32.74}$\pm$\textbf{2.94}              &\textbf{0.835}$\pm$\textbf{0.063}              &\textbf{38.12}$\pm$\textbf{2.51}              &\textbf{ 0.908}$\pm$\textbf{0.048}              &\textbf{36.50}$\pm$\textbf{2.59}              &\textbf{0.885}$\pm$\textbf{0.056}              \\
\hline
\end{tabular}
\end{adjustbox}
\label{tb1}
\vspace{-4pt}
\end{table*}
\begin{figure*}[!t]
\small
\centering
\begin{adjustbox}{width=1\linewidth}
\begin{tikzpicture}
\draw (0, 0) node[inner sep=0] {\includegraphics[width=1\textwidth]{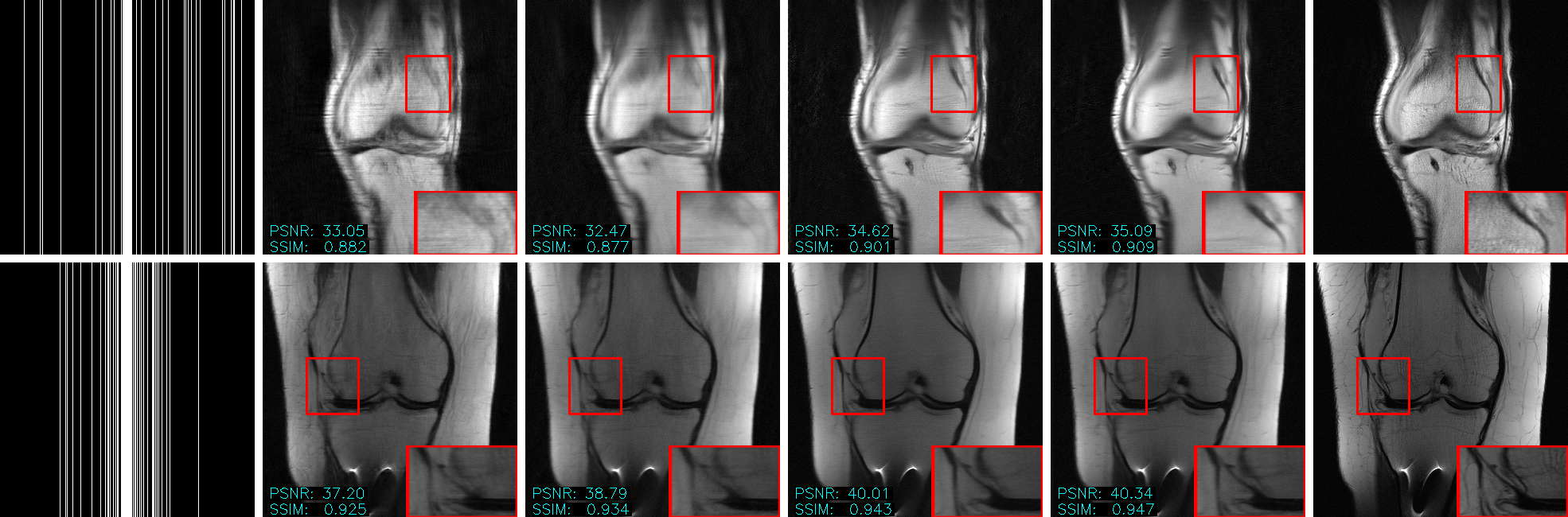}};
\draw (-7.5, 3.15) node {\fontsize{8pt}{9pt}\selectfont Mask};
\draw (-4.4, 3.15) node {\fontsize{8pt}{9pt}\selectfont DPS};
\draw (-1.6, 3.15) node {\fontsize{8pt}{9pt}\selectfont DDNM};
\draw (1.6, 3.15) node {\fontsize{8pt}{9pt}\selectfont {BGDM} (\textbf{ours})};
\draw (4.4, 3.15) node {\fontsize{8pt}{9pt}\selectfont {R-BGDM} (\textbf{ours})};
\draw (7.5, 3.15) node {\fontsize{8pt}{9pt}\selectfont References}; 
\draw (-9.25, 1.6) node[rotate=90] {\fontsize{7pt}{7pt}\selectfont Uniform1D};
\draw (-9.25, -1.6) node[rotate=90] {\fontsize{7pt}{7pt}\selectfont Gaussian1D};
\end{tikzpicture}
\end{adjustbox}
\caption{The visual representation of results from the fastMRI knee dataset, obtained using 100 steps for Gaussian1D and Uniform1D masks at an ACR of 8.}
\label{fig5}
\vspace{-6pt}
\end{figure*}
\subsection{Setup}
\subsubsection{Data Sets}
to demonstrate the performance of our proposed method, we present our image reconstruction evaluation on four publicly available datasets. For undersampled MRI experiments, we rely on real-valued Brain Tumor Segmentation (BraTS) 2021 \cite{menze2014multimodal, bakas2017advancing} and complex-valued fastMRI knee datasets \cite{zbontar2018fastmri}. In our evaluation with the BraTS dataset, we follow the approach outlined in \cite{song2021solving}, where 3D MRI volumes are sliced to obtain 297,270 images with an image size of $240\times240$ for the training set. We simulate MRI measurements using the Fast Fourier Transform (FFT) and undersample the k-space using an equispaced Cartesian mask, from an acceleration factor of 4 to 24. When conducting experiments on fastMRI knee dataset, we follow \cite{chung2022score} to appropriately crop the raw k-space data to 320$\times$320 pixels. We then generate single-coil minimum variance unbiased estimator images as our ground truth references. To simulate measurements, the data is processed using the FFT and then undersampled with one-dimensional Gaussian and Uniform masks with acceleration factors 4 and 8. For the sparse-view CT reconstruction experiment, we used the Lung Image Database Consortium (LIDC) dataset \cite{armato2011lung, clark2013cancer}. From this dataset, we derived 130,304 two-dimensional images with the size of $320\times320$ by slicing the original 3D CT volumes. We produce sinograms using a parallel-beam setup with evenly spaced projection angles of 10 and 23 over 180 degrees, simulating sparse-view CT acquisitions. For the super-resolution task, we utilized the fastMRI brain datasets by downsampling the full-resolution 2D images retrospectively. We selected about 63\% of all 2D images labeled as `reconstruction rss', resulting in 34,698 brain slices for training.
\vspace{1pt}
\subsubsection{Baselines}
we primarily compare our proposed method with two SOTA zero-shot inverse problem solvers: DPS \cite{chung2022diffusion} and DDNM \cite{wang2022zero}. For the Knee fastMRI dataset, we reported the result of Score-MRI \cite{chung2022score} directly from their paper. To ensure a fair comparison, we adopt the incorporation strategies from these methods, along with appropriate parameter settings within our architecture. Also, for CT reconstruction, we replaced the DPS with ScoreMed \cite{song2021solving}. In our experiments, it was observed that the recurrent use of Filtered Back Projection (FBP) tends to be numerically unstable in DPS, frequently resulting in overflow. This has also been reported in \cite{chung2022improving}. Accordingly, for CT reconstruction, we consider ${\color{orange}\zeta}=0$ in Alg \ref{alg1}. For all experiments, results are reported in terms of peak signal-to-noise ratio (PSNR) and structural similarity (SSIM) metrics on a dataset of 1,000 test images. 
\vspace{-2pt}
\subsection{Implementation}\label{A5}
\subsubsection{Architecture}
to learn the prior, we train diffusion-based generative networks using the ADM architecture \cite{dhariwal2021diffusion} and the default parameters presented in \cite{song2021solving}. The models are trained with classifier-free
diffusion guidance without dropout probability. We develop distinct networks for distinct medical imaging tasks: one for undersampled reconstruction using the real-valued BraTS dataset, another for similar tasks with the complex-valued fastMRI dataset, a third for sparse-view reconstruction on the LIDC-CT dataset, and a fourth for super-resolution using the fastMRI brain dataset. 

\begin{figure*}[!t]
\small
\centering
\begin{adjustbox}{width=1\linewidth}
\begin{tikzpicture}
\draw (0, 0) node[inner sep=0] {\includegraphics[width=1\textwidth]{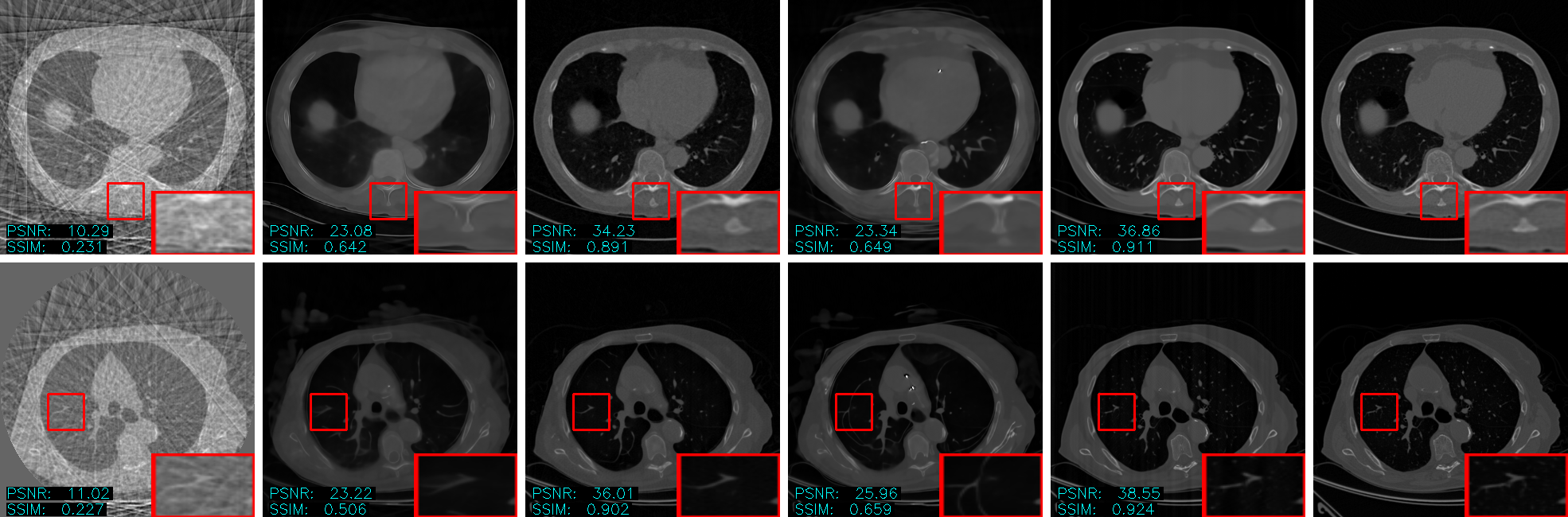}};
\draw (-7.5, 3.15) node {\fontsize{8pt}{9pt}\selectfont FBP};
\draw (-4.4, 3.15) node {\fontsize{8pt}{9pt}\selectfont DDNM};
\draw (-1.6, 3.15) node {\fontsize{8pt}{9pt}\selectfont ScoreMed};
\draw (1.6, 3.15) node {\fontsize{8pt}{9pt}\selectfont {BGDM} (\textbf{ours})};
\draw (4.4, 3.15) node {\fontsize{8pt}{9pt}\selectfont {R-BGDM} (\textbf{ours})};
\draw (7.5, 3.15) node {\fontsize{8pt}{9pt}\selectfont References}; 
\end{tikzpicture}
\end{adjustbox}
\caption{Examples of sparse-view CT reconstruction results on LIDC, all with 23 projections.}
\label{fig6}
\vspace{-6pt}
\end{figure*}
\begin{figure*}[!t]
\small
\centering
\begin{adjustbox}{width=1\linewidth}
\begin{tikzpicture}
\draw (0, 0) node[inner sep=0] {\includegraphics[width=1\textwidth]{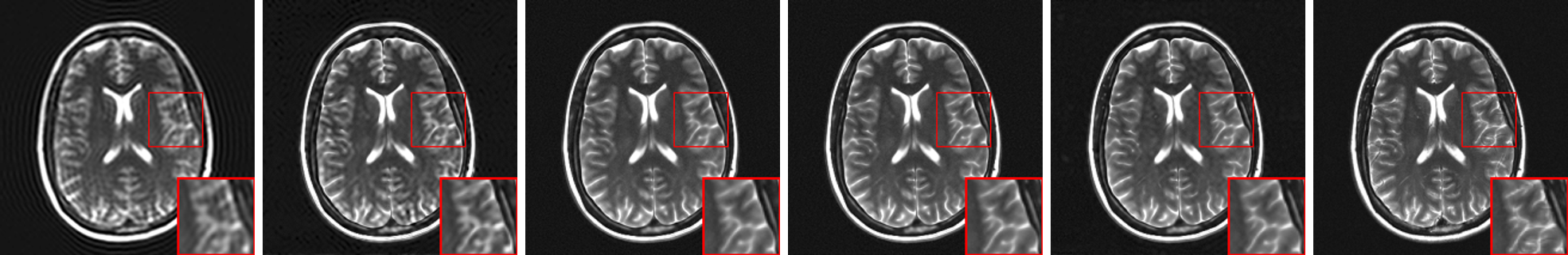}};
\draw (-7.4, 1.7) node {\fontsize{8pt}{9pt}\selectfont Low-resolution};
\draw (-4.5, 1.7) node {\fontsize{8pt}{9pt}\selectfont DPS};
\draw (-1.6, 1.7) node {\fontsize{8pt}{9pt}\selectfont DDNM};  
\draw (1.6, 1.7) node {\fontsize{8pt}{9pt}\selectfont {{BGDM} (\textbf{ours})}};
\draw (4.5, 1.7) node {\fontsize{8pt}{9pt}\selectfont {{R-BGDM} (\textbf{ours})}};
\draw (7.4, 1.7) node {\fontsize{8pt}{9pt}\selectfont References}; 
\end{tikzpicture}
\end{adjustbox}
\caption{The qualitative super-resolution results of the fastMRI brain dataset at ACR 16 with 50 steps.}
\label{fig7}
\end{figure*}
\subsubsection{Sampling Settings and Hyper-parameters.}
For a sampling of our method, we fine-tune parameters $({\color{orange}\zeta}, {\color{orange}\lambda}, {\color{orange}\gamma})$ using cross-validation for each datasets. We observed that extensive hyper-parameter tuning is not required to obtain top-performance results. Accordingly, we limit the hyper-parameter search for each task to $\color{orange}\zeta$ $\in [0 \hspace{3pt} 2]$, $\color{orange}\gamma$ $\in [0 \hspace{3pt} 4.5]$, and $\color{orange}\lambda$ $\in [{10}^{-3} \hspace{3pt} {10}^{-4}]$. For other methods of DPS and DDNM, we relied on their original implementation. In the experiments conducted on the BraTS and fastMRI Knee datasets, we set the parameter $\eta$ to 0.85. Conversely, for the LIDC-CT experiment and fastMRI brain dataset, we assign $\eta$ a value of 0.95 and 1 respectively. Regarding the sampling process, we consider over 200 timesteps for the BraTS dataset, 100 timesteps for fastMRI, and 350 timesteps for the LIDC-CT dataset. All these experiments are performed on a single 3090 NVIDIA GPU with a batch size of 1.
\begin{table}[!t]
\captionsetup{font=small, aboveskip=0pt, belowskip=0pt}
\caption{Quantitative results of sparse-view CT reconstruction on the LIDC dataset with 350 NFEs. Methods marked with * indicate that the parameter {\color{orange}$\zeta$} is set to zero.}
\setlength{\tabcolsep}{4pt} 
\renewcommand{\arraystretch}{0.8}
\centering
{\setlength{\extrarowheight}{1.5pt}
\begin{adjustbox}{max width=\textwidth}
\begin{tabular}{c|cc|cc}
\Xhline{2\arrayrulewidth}
\multirow{2}{*}[-1pt]{Method}  & \multicolumn{2}{c|}{\textbf{23}  projection} & \multicolumn{2}{c}{ \textbf{10}  projection} \\\cline{2-5}
& PSNR$\uparrow$ & SSIM$\uparrow$ & PSNR$\uparrow$ & SSIM$\uparrow$ \\
\hline
FBP & 10.07\scalebox{0.7}{$\pm$1.40} & 0.218\scalebox{0.7}{$\pm$0.070} & -- & -- \\
DDNM \cite{wang2022zero}& 23.76\scalebox{0.7}{$\pm$2.21} & 0.624\scalebox{0.7}{$\pm$0.077} & 18.35\scalebox{0.7}{$\pm$2.30} & 0.696\scalebox{0.7}{$\pm$0.047} \\
ScoreMed \cite{song2021solving}& 35.24\scalebox{0.7}{$\pm$2.71} & 0.905\scalebox{0.7}{$\pm$0.046} & 29.52\scalebox{0.7}{$\pm$2.63} & 0.823\scalebox{0.7}{$\pm$0.061} \\ \hline
\rowcolor{LightOrange}
BGDM* (\textbf{ours}) & 25.89\scalebox{0.7}{$\pm$2.43} & 0.671\scalebox{0.7}{$\pm$0.069} & 20.14\scalebox{0.7}{$\pm$2.35} & 0.723\scalebox{0.7}{$\pm$0.043} \\ \rowcolor{mylightblue}
R-BGDM* (\textbf{ours}) & \textbf{35.82}\scalebox{0.7}{\textbf{$\pm${2.45}}} & \textbf{0.911}\scalebox{0.7}{\textbf{$\pm$0.052}} & \textbf{30.22}\scalebox{0.7}{$\pm${\textbf{2.48}}} & \textbf{0.834}\scalebox{0.7}{\textbf{$\pm${0.056}}} \\
\Xhline{2.2\arrayrulewidth}
\end{tabular}
\end{adjustbox}
}\label{tb2}
\vspace{-6pt}
\end{table}
\begin{table}[!t]
\captionsetup{font=small, aboveskip=0pt, belowskip=0pt}
\caption{Super-resolution results on fastMRI Brain.}
\setlength{\tabcolsep}{4pt} 
\renewcommand{\arraystretch}{0.8}
\centering
{\setlength{\extrarowheight}{1.5pt}
\begin{adjustbox}{max width=\textwidth}
\begin{tabular}{c|cc|cc}
\Xhline{2\arrayrulewidth}
\multirow{2}{*}[-1pt]{Method}  & \multicolumn{2}{c|}{\textbf{2$\times$2}  SR} & \multicolumn{2}{c}{ \textbf{4$\times$4}  SR} \\\cline{2-5}
& PSNR$\uparrow$ & SSIM$\uparrow$ & PSNR$\uparrow$ & SSIM$\uparrow$ \\
\hline
DPS \cite{chung2022diffusion}  & 35.44\scalebox{0.7}{$\pm$3.71} & 0.931\scalebox{0.7}{$\pm$0.027} & 30.29\scalebox{0.7}{$\pm$2.84} & 0.854\scalebox{0.7}{$\pm$0.034} \\
DDNM \cite{wang2022zero} & 36.12\scalebox{0.7}{$\pm$3.88} & 0.947\scalebox{0.7}{$\pm$0.020} & 31.84\scalebox{0.7}{$\pm$2.95} & 0.866\scalebox{0.7}{$\pm$0.026} \\ \hline
\rowcolor{LightOrange}
BGDM (\textbf{ours}) & 36.18\scalebox{0.7}{$\pm$3.89} & 0.946\scalebox{0.7}{$\pm$0.020} & 32.08\scalebox{0.7}{$\pm$2.96} & 0.868\scalebox{0.7}{$\pm$0.026} \\ \rowcolor{mylightblue}
R-BGDM (\textbf{ours}) & \textbf{36.23}\scalebox{0.7}{$\pm${\textbf{3.87}}} & \textbf{0.946}\scalebox{0.7}{\textbf{$\pm${0.020}}} & \textbf{32.18}\scalebox{0.7}{\textbf{$\pm$2.97}} & \textbf{0.876}\scalebox{0.7}{$\pm$0.026} \\
\Xhline{2.2\arrayrulewidth}
\end{tabular}
\end{adjustbox}
}\label{tb3}
\vspace{-12pt}
\end{table}
\subsection{Results}
Figure \ref{fig3} illustrates a comparative analysis of the reconstruction quality across a dataset of 1000 BraTS images, using metrics such as PSNR and SSIM, to evaluate performance at various Acceleration Rates (ACR) and with different Network Function Evaluators (NFEs). The evaluation underscores the superior performance of both BGDM over other methods, demonstrating not only higher accuracy but also efficiency in computational time. Notably, BGDM, even at a modest 100 NFEs, significantly performs better than DDNM and DPS operating at a substantially higher 350 NFEs, establishing its noteworthy efficacy in producing accurate reconstructions swiftly. In Figure \ref{fig4}, we display the BraTS image reconstruction results using different methods for test measurements undersampled at 4, 8, and 24 acceleration factors. Our BGDM method achieves superior image fidelity, preserving lesion heterogeneities at 4x and 8x undersampling levels. Unlike DDNM and DPS, both BGDM and R-BGDM maintain data fidelity even at 24x undersampling, producing highly consistent images with the ground truth. The comparison of various methods on the fastMRI knee dataset with 100 NFEs is presented in Table \ref{tb1}, with an illustrative case for each sampling mask type (i.e., Uniform1D and Gaussian1D) showcased in Figure \ref{fig5}. DPS failed to reconstruct acceptable images due to the short 100 sampling steps and DDNM displayed noticeable residual aliasing artifacts. Notably, our BGDM method demonstrated superior performance compared to Score-MRI \cite{chung2023fast} and DDNM by a margin of \textbf{2dB} and \textbf{1.5dB}, respectively. Table \ref{tb2} shows the average results from 1000 test CT images using both 23 and 10 projections. Given ${\color{orange}\zeta}=0$ in Alg \ref{alg1}, R-BGDM  slightly outperforms ScoreMed in terms of PSNR and SSIM values, with both significantly surpassing DDNM. Figure \ref{fig6} illustrates the results of reconstructing two CT lung images from 23 projections using multiple methods. Our method recovers finer details, as seen in the zoomed-in views, and achieves the highest PSNR and SSIM values. The average results of the 2$\times$2 and 4$\times$4 super-resolution tasks on 1000 images are reported in Table \ref{tb3}, with BGDM and R-BGDM slightly over-performing DDNM. Figure \ref{fig7} demonstrates the 4$\times$4 super-resolution result from different methods. DPS showed severe residual ringing artifacts, while our BGDM image successfully recovered the detailed brain structures with less blurring than DDNM.
\begin{table*}[!t]
\small
\captionsetup{font=small, aboveskip=0pt, belowskip=0pt}
\caption{Ablation study results for undersampled MRI reconstruction using the BraTS dataset.}
\setlength{\tabcolsep}{4pt} 
\renewcommand{\arraystretch}{0.8}
\centering
{\setlength{\extrarowheight}{1.5pt}
\begin{adjustbox}{max width=\textwidth}
\begin{tabular}{c|cc|cc|cc}
\Xhline{2.2\arrayrulewidth}
\multirow{2}{*}[-1pt]{Method}  & \multicolumn{2}{c|}{4$\times$ ACR} & \multicolumn{2}{c|}{8$\times$ ACR} & \multicolumn{2}{c}{24$\times$ ACR} \\\cline{2-7}
& PSNR$\uparrow$ & SSIM$\uparrow$ & PSNR$\uparrow$ & SSIM$\uparrow$ & PSNR$\uparrow$ & SSIM$\uparrow$\\
\hline
DPS \cite{chung2022diffusion}  & 37.84\scalebox{0.7}{$\pm$2.26} & 0.948\scalebox{0.7}{$\pm$0.018} & 35.98\scalebox{0.7}{$\pm$2.15} & 0.939\scalebox{0.7}{$\pm$0.020} & 29.46\scalebox{0.7}{$\pm$3.66} & 0.815\scalebox{0.7}{$\pm$0.067}\\
DDNM \cite{wang2022zero}& 39.92\scalebox{0.7}{$\pm$2.35} & 0.965\scalebox{0.7}{$\pm$0.012} & 35.18\scalebox{0.7}{$\pm$2.10} & 0.940\scalebox{0.7}{$\pm$0.017} & 27.09\scalebox{0.7}{$\pm$2.94} & 0.841\scalebox{0.7}{$\pm$0.049}\\
Ours\,${}_{\text{no-\textbf{or}}}$ & 32.38\scalebox{0.7}{$\pm$1.89} & 0.874\scalebox{0.7}{$\pm$0.030} &  29.56\scalebox{0.7}{$\pm$2.01} & 0.845\scalebox{0.7}{$\pm$0.034} & 23.16\scalebox{0.7}{$\pm$2.53} & 0.794\scalebox{0.7}{$\pm$0.044}\\
Ours\,$_{\text{no-\textbf{ir}}}$ & 39.97\scalebox{0.7}{$\pm$2.31} & 0.969\scalebox{0.7}{$\pm$0.011} & {35.36\scalebox{0.7}{$\pm$2.03}} & {0.943\scalebox{0.7}{$\pm$0.015}} & 27.36\scalebox{0.7}{$\pm$2.78} & 0.849\scalebox{0.7}{$\pm$0.041}\\ 
Ours\,$_{\text{no-\textbf{i}}}$ & 41.37\scalebox{0.7}{$\pm$2.72} & 0.967\scalebox{0.7}{$\pm${0.009}} & 37.06\scalebox{0.7}{$\pm$2.04} & 0.923\scalebox{0.7}{$\pm$0.011} & 28.37\scalebox{0.7}{$\pm$3.23} & 0.832\scalebox{0.7}{$\pm$0.047}\\\hline
\rowcolor{LightOrange}
BGDM (\textbf{ours}) & 41.54\scalebox{0.7}{$\pm$2.90} & \textbf{0.980}\scalebox{0.7}{$\pm$\textbf{0.008}} & 38.02\scalebox{0.7}{$\pm$2.31} & 0.961\scalebox{0.7}{$\pm${0.009}} & 29.87\scalebox{0.7}{$\pm$3.31} & {0.887}\scalebox{0.7}{$\pm${0.036}}\\
\rowcolor{mylightblue}
{R-BGDM} (\textbf{ours}) & \textbf{41.94}\scalebox{0.7}{\textbf{$\pm$2.88}} & {0.977\scalebox{0.7}{$\pm${0.008}}} & \textbf{38.46}\scalebox{0.7}{\textbf{$\pm${2.54}}} & \textbf{0.964}\scalebox{0.7}{\textbf{$\pm${0.011}}} & \textbf{30.04}\scalebox{0.7}{\textbf{$\pm$3.33}} & \textbf{0.887}\scalebox{0.7}{\textbf{$\pm$0.039}}\\
\Xhline{2.2\arrayrulewidth}
\end{tabular}
\end{adjustbox}
}\label{tb4}
\vspace{-8pt}
\end{table*}
\begin {figure*}[!t]
\centering
\begin{adjustbox}{width=1\linewidth}
\begin{tikzpicture}
\draw (0, 0) node[inner sep=0] {\includegraphics[width=1\textwidth, height=3cm]{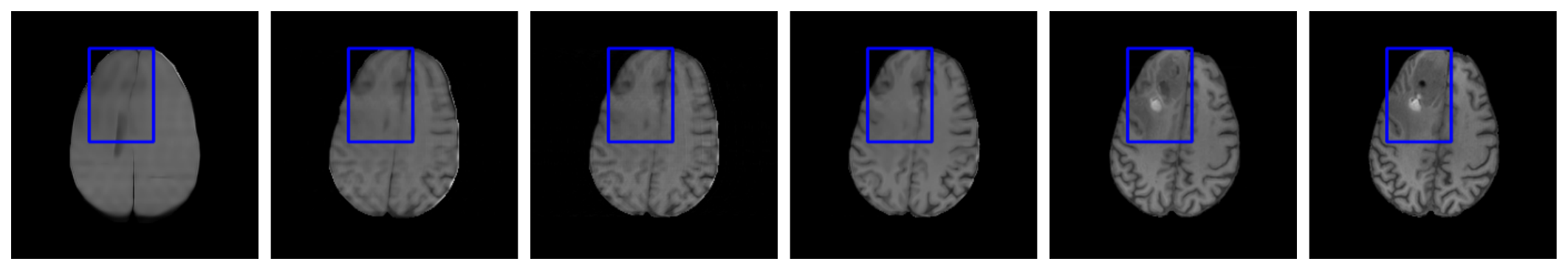}};
\draw (-7.5, 1.6) node {\fontsize{8pt}{9pt}\selectfont Ours\,$_{\text{no-\textbf{or}}}$};
\draw (-4.5, 1.6) node {\fontsize{8pt}{9pt}\selectfont Ours\,$_{\text{no-\textbf{ir}}}$};
\draw (-1.5, 1.6) node {\fontsize{8pt}{9pt}\selectfont Ours\,$_{\text{no-\textbf{i}}}$};
\draw (1.5, 1.6) node {\fontsize{8pt}{9pt}\selectfont {BGDM} };
\draw (4.5, 1.6) node {\fontsize{8pt}{9pt}\selectfont {R-BGDM}};
\draw (7.5, 1.6) node {\fontsize{8pt}{9pt}\selectfont Reference};  
\end{tikzpicture}
\end{adjustbox}
\caption{A representative visual result of the ablation study, showcasing the 24x scenario.}
\label{fig8}
\vspace{-8pt}
\end{figure*}
\subsection{Ablation Studies}
To assess the impact of key components in \textbf{R-BGDM} sampling algorithm, we performed ablations on the undersampled MRI task using the BraTS dataset. The summarized outcomes are presented in Table \ref{tb3}, evaluating three key variations in Algorithm \ref{alg1}: (i) the exclusion of\underline{ \emph{\textbf{o}uter-level} optimization} (step \textbf{\ref{step9}}) and additional \underline{\textbf{r}efinement gradient} step termed `no-\textbf{or}', (ii) the omission of \underline{ \emph{\textbf{i}nner-level} guidance} (step \textbf{\ref{step8}}) and \underline{\textbf{r}efinement gradient} designated as `no-\textbf{ir}', and (iii) the absence of \underline{ \emph{\textbf{i}nner-level} guidance} alone (step \textbf{\ref{step8}}) noted as `no-\textbf{i}'. Our observations indicate that proximal optimization plays the most substantial role in both BGDM and R-BGDM methods, with improvements achieved through { \emph{\textbf{i}nner-level} guidance} and further refinement. Remarkably, our algorithm outperforms all baselines even without the refinement step, yet further improves performance when this step is incorporated, as illustrated in Fig \ref{fig8}.
\section{Conclusion, Limitation, and Future Works}
\vspace{-2pt}
In this paper, we propose an effective bi-level guided conditional sampling approach for diffusion models to tackle inverse problems in medical imaging. Through extensive experiments, our method demonstrates its superiority to other methods on several highly heterogeneous, publicly available medical datasets, thereby validating our analysis. Theoretically, our approach is amenable to resolving other linear inverse problems such as inpainting, super-resolution, deblurring, and so forth, provided that the pertinent diffusion model is accessible. 
A few limitations remain that deserve further examination.
\begin{itemize}

\item Despite achieving superior reconstruction results compared to other methods \cite{song2021solving, chung2022diffusion, wang2022zero} and demonstrating more efficient sampling for medical imaging applications \cite{chung2022score, song2021solving, jalal2021robust, chung2023fast}, BGDM remains sensitive to hyperparameters. Therefore, exploring a more general hyperparameter tuning approach, such as Bayesian optimization, would be beneficial. 

\item {{It should be noted that our CT simulation adheres to the 2D parallel beam geometry assumption, aligning with the baseline models used in other studies for direct comparison. This differs from the more complex and realistic 3D cone-beam CT or helical CT simulations \cite{kim2014combining}. Additionally, the BraTS dataset, employed both in our study and by the baseline methods, has been indicated in a recent paper \cite{shimron2022implicit} to have an overestimated undersampling factor, which arises from the conjugate symmetry of k-space inherent in real-valued images.}}
\end{itemize}
In future work, we plan to enhance our model for compatibility with 3D simulations and adaptability to distributional shifts \cite{barbano2023steerable, askari2023mapflow}.

\balance
\bibliography{references}
\bibliographystyle{unsrt}

\clearpage
\balance
\section*{Appendices}
\begin{appendices}

\subsection{Related Works}\label{A1}
A solution to the inverse problem $\mathbf{y}=\mathcal{A}\mathbf{x}+\mathbf{n}$,  can be probabilistically derived via the maximum likelihood estimation (MLE), defined as ${\mathbf{x}}_{\text{ML}} = \arg\max_{\mathbf{x}} \log p(\mathbf{y}|\mathbf{x})$, where $\textstyle{p(\mathbf{y}|\mathbf{x})}:= \mathcal{N}(\mathcal{A}\mathbf{x}, \sigma_{\mathbf{y}}^2)$ represents the likelihood of observation $\textstyle{\mathbf{y}}$, ensuring data consistency. Nevertheless, if the forward operator $\textstyle{\mathcal{A}}$ is singular, e.g., when $m < n$, the problem is ill-posed. In such cases, it is fundamentally infeasible to uniquely recover the signal set $\textstyle{\mathcal{X}}$ using only the observed measurements $\textstyle{\mathcal{Y}}$, even in the noiseless scenario where $\textstyle{\mathcal{Y} = \mathcal{A}\mathcal{X}}$. This challenge arises due to the nontrivial nature of the null space of $\textstyle{\mathcal{A}}$. 

To mitigate the ill-posedness, it is therefore essential to incorporate an additional assumption based on \emph{prior} knowledge to constrain the space of possible solutions. A predominantly adopted framework that offers a more meaningful solution is Maximum a Posteriori (MAP) estimation which is formulated as $\textstyle{{\mathbf{x}}_{\text{MAP}} = \arg\max_{\mathbf{x}} [{\log p(\mathbf{y}|\mathbf{x})} + \log p(\mathbf{x})]}$, where the term $\textstyle{\log p(\mathbf{x})}$ encapsulates the prior information of the clean image $\textstyle{\mathbf{x}}$. 

The concept of priors in solving inverse problems has evolved considerably over time. Classically, many methodologies relied on hand-crafted priors, which are analytically defined constraints such as sparsity \cite{candes2008introduction, tang2009performance}, low-rank \cite{fazel2008compressed, cui2014likelihood}, total variation \cite{candes2006robust}, to name but a few, to enhance reconstruction. With the advent of deep learning models, priors have transitioned to being data-driven, yielding significant gains in reconstruction quality \cite{bora2017compressed, mardani2018deep, ardizzone2018analyzing, goh2019solving, asim2020invertible, whang2021solving}. These priors, whether learned in a supervised or unsupervised fashion, have been integrated within the MAP framework to address ill-posed inverse problems. In the supervised paradigm, the reliance on the availability of paired original images and observed measurements also can potentially limit the model's generalizability. As a result, the trend has shifted towards an increasing interest in unsupervised approaches, where priors are learned implicitly or explicitly using deep generative models. 

The strategies within the unsupervised learning paradigm vary based on how the learned priors (a.k.a. generative priors) are imposed. For instance,  generators \( \mathcal{G}_{\theta} \) in pre-trained generative models such as Generative Adversarial Networks (GANs)\cite{goodfellow2016deep, bora2017compressed}, Variational Autoencoders (VAEs) \cite{ardizzone2018analyzing}, and Normalizing Flows (NFs) \cite{asim2020invertible}, are employed as priors to identify the latent code that explains the measurements, as described by the optimization problem
$\hat{\mathbf{z}} = \arg\max_{\mathbf{z}} \, \log p(\mathbf{y}|\mathcal{G}_{\theta}(\mathbf{z})) + \log p(\mathbf{z})$. In such a way, the solution $\hat{\mathbf{z}}$ is constrained to be within the domain of the generative model. This approach, however, suffers from critical restrictions. In the first place, the low dimensionality of the latent space is a major concern, as it hampers the reconstruction of images that lie outside their manifold. Additionally, it demands computationally expensive iterative updates, given the complexity of generator $\mathcal{G}_{\theta}$. Crucially, the deterministic nature of the recovered solutions hinders the assessment of the reliability of reconstruction. In fact, MAP inference fails to fully capture the entire range of the solution spectrum, particularly when solving an ill-posed problem that might hold multiple solutions aligned closely with both the observed measurements and prior assumptions.

To account for the variety within the solution domain and to measure reconstruction certainty, the inverse problem is tackled from a Bayesian inference standpoint. Bayesian inference yields a posterior distribution, $p(\mathbf{x}|\mathbf{y})$, from which multiple conditional samples can be extracted \cite{brooks2011handbook, blei2017variational}. Recently, pre-trained diffusion models \cite{ho2020denoising, nichol2021improved} are utilized as a powerful generative prior (a.k.a denoiser), in a zero-shot manner, to effectively sample from the conditional posterior \cite{kadkhodaie2021stochastic, daras2022score, rombach2022high}. The strategies for posterior (conditional) sampling via diffusion models fall into two distinct approaches. In the first approach, diffusion models are trained conditionally, directly embedding the conditioning information $\mathbf{y}$ during the training phase \cite{ho2020denoising, rombach2022high, liu20232}. However, conditional training tends to require: (i) the assembly of a massive amount of paired data and its corresponding conditioning ($\mathbf{x}, \mathbf{y}$), and (ii) retraining when testing on new conditioning tasks, highlighting the adaptability issues. In the second approach,  unconditionally pre-trained diffusion models are employed as generative prior (a.k.a denoiser) to perform conditional sampling for certain tasks. A primary difficulty, however, is how to impose data consistency between measurements and the generated images in each iteration \cite{chung2022diffusion, wang2022zero, chung2023fast}.

\subsection{Closed-form solutions}\label{A3}

Consider the following optimization problem in \eqref{eq10}

\begin{align*}
{\hat{\mathbf{x}}_{0|t}} = \underset{\mathbf{x}}{\operatorname{arg\, min}} \frac{1}{2}\Vert\mathbf{y}-\mathcal{A}\mathbf{x}\Vert_{2}^{2} + \frac{\lambda}{2} \Vert \mathbf{x}-{\mathbf{x}_{0|t}} \Vert_{2}^{2}.
\end{align*}
For the MRI reconstruction task, we express  $\mathcal{A}\mathbf{x}  = \mathcal{M} \odot (\mathcal{F}\mathbf{x}) = \mathcal{M} \odot \mathbf{w}$, where \(\mathcal{M}\) represents the Cartesian equispaced mask, \(\mathcal{F}\) is the Fourier matrix, and \(\odot\) signifies element-wise multiplication. 
Given this definition, and considering the identity $\textstyle{{\operatorname{arg\, min}}_{\mathbf{x}}\Vert \mathbf{x}-{\mathbf{x}_{0|t}} \Vert_{2}^{2}={\operatorname{arg\, min}}_{\mathbf{x}}\Vert \mathcal{F}\mathbf{x}-\mathcal{F}{\mathbf{x}_{0|t}} \Vert_{2}^{2}}$, then the optimization problem in terms of \(\mathbf{w}\) can be redefined as
\[ {\hat{\mathbf{w}}_{0|t}} = \underset{\mathbf{w}}{\operatorname{arg\, min}} \frac{1}{2}\Vert\mathcal{M} \odot \mathbf{w}-\mathbf{y}\Vert_{2}^{2} + \frac{\lambda}{2}\Vert\mathbf{w}-{\mathbf{w}_{0|t}}\Vert_{2}^{2}. \]
By expanding the L2-norm terms, we obtain
\begin{align*} {\hat{\mathbf{w}}_{0|t}} = \underset{\mathbf{w}}{\operatorname{arg\, min}} \sum_{i=1}^n (m_iw_i-y_i)^2 + \lambda \sum_{i=1}^n (w_i-w_{0|t}^i)^2. 
\end{align*}
The solution for \({\hat{\mathbf{w}}_{0|t}}\) is
\begin{align*}
{\hat{\mathbf{w}}_{0|t}} = \frac{\mathcal{M} \mathbf{y} + \lambda \mathbf{w}_{0|t}}{\mathcal{M}+\lambda}.
\end{align*}
Given the relation ${\hat{\mathbf{x}}_{0|t}} = \mathcal{F}^{-1}{\hat{\mathbf{w}}_{0|t}}$, we can then deduce
\begin{equation*}
{{\hat{\mathbf{x}}_{0|t}} = \mathcal{F}^{-1}\left(\frac{\mathcal{M} \mathbf{y} + \lambda \mathcal{F}\mathbf{x}_{0|t}}{\mathcal{M}+\lambda}\right)}
\end{equation*}
Consider the following range-null space decomposition defined in \eqref{eq8}
\begin{align*}
{\hat{\mathbf{x}}_{0|t}} = \mathcal{A}^{\dagger} \mathbf{y}+\left(\mathbf{I}-\mathcal{A}^{\dagger} \mathcal{A}\right) \mathbf{x}_{0|t}.
\end{align*}
where \(\mathcal{A}^{\dagger}\) denotes the pseudo-inverse of matrix \(\mathcal{A}\) and \(\mathbf{I}\) is the identity matrix. 

For MRI, the forward operator is modelled as \(\mathcal{A}=\mathcal{M} \mathcal{F}\). An important property that arises is \(\mathcal{A}\mathcal{A}\mathcal{A}\equiv\mathcal{A}\), which suggests that \(\mathcal{A}\) itself can be represented as its pseudo-inverse \(\mathcal{A}^{\dagger}\). With this property, the pseudo-inverse is then expressed as \(\mathcal{A}^{\dagger}=\mathcal{F}^{-1} \mathcal{M}\). Substituting this representation into our original expression, we obtain
\[ {\hat{\mathbf{x}}_{0|t}} = \mathcal{F}^{-1} \mathcal{M} \mathbf{y}+\left(\mathbf{I}-\mathcal{F}^{-1} \mathcal{M} \mathcal{F}\right) \mathbf{x}_{0|t}. \]
Using the Fourier identity \(\mathcal{F}^{-1} \mathcal{F} = \mathbf{I}\), we can further simplify this to:
\[{{\hat{\mathbf{x}}_{0|t}} = \mathcal{F}^{-1}\left(\mathcal{M} \mathbf{y}+(\mathbf{I}-\mathcal{M}) \mathcal{F} \mathbf{x}_{0|t}\right)} \]

\subsection{Posterior mean}\label{A4}

\subsubsection{Posterior mean with additional measurement for VPSDE }\label{A41}
A notable SDE with an analytic transition probability is the variance-Preserving SDE (VPSDE) \cite{song2020score, karras2022elucidating}, which considers $\textstyle{\mathbf{f}(\mathbf{x}_{t},t)=-\frac{1}{2}\beta(t)\mathbf{x}_{t}}$ and $\textstyle{g(t)=\sqrt{\beta(t)}}$, where $\textstyle{\beta(t)=\beta_{min}+t(\beta_{max}-\beta_{min})}$; and its transition probability follows a Gaussian distribution of $\textstyle{p_{0t}(\mathbf{x}_{t}\vert\mathbf{x}_{0})=\mathcal{N}(\mathbf{x}_{t};\bm{\mu}_t\mathbf{x}_{0},\bm{\sigma}_t^{2}\mathbf{I})}$ with $\textstyle{\bm{\mu}_t=\exp\{-\frac{1}{2}\int_{0}^{t}\beta(s)\mathrm{s}}\}$ and $\bm{\sigma}_t^{2}=1-\exp\{-\int_{0}^{t}\beta(s)\mathrm{s}\}$. Given such transition probability, we seek to derive the corresponding posterior mean with additional measurement.

Begin by representing the distribution \( p(\mathbf{x}_t | \mathbf{y}) \) as marginalizing out \( \mathbf{x}_0 \) conditioned on \( \mathbf{y} \):
\begin{align*}
p(\mathbf{x}_t | \mathbf{y}) = \int_{\mathbf{x}_0} p(\mathbf{x}_t | \mathbf{x}_0, \mathbf{y}) p(\mathbf{x}_0 | \mathbf{y}) d\mathbf{x}_0.    
\end{align*}
Differentiate w.r.t. \( \mathbf{x}_t \) on both sides
\begin{align*}
\nabla_{\mathbf{x}_t} p(\mathbf{x}_t | \mathbf{y}) = \int_{\mathbf{x}_0} p(\mathbf{x}_0 | \mathbf{y}) \nabla_{\mathbf{x}_t} p(\mathbf{x}_t | \mathbf{x}_0, \mathbf{y}) d\mathbf{x}_0.
\end{align*}
With our new probability distribution model, the gradient becomes
\begin{align*} 
\nabla_{\mathbf{x}_t} \log p(\mathbf{x}_t | \mathbf{x}_0) = \frac{(\bm{\mu}_t \mathbf{x}_0-\mathbf{x}_t)}{\bm{\sigma}^2_t}.
\end{align*}
Inserting this into our previous equation, we have
\begin{align*}
\nabla_{\mathbf{x}_t} p(\mathbf{x}_t | \mathbf{y}) = \int_{\mathbf{x}_0} p(\mathbf{x}_0 | \mathbf{y}) p(\mathbf{x}_t | \mathbf{x}_0, \mathbf{y}) \frac{(\bm{\mu}_t \mathbf{x}_0-\mathbf{x}_t)}{\bm{\sigma}^2_t} d\mathbf{x}_0.
\end{align*}
Simplifying the above equation, we get:
\begin{align*}
\nabla_{\mathbf{x}_t} p(\mathbf{x}_t | \mathbf{y}) &= \frac{1}{\bm{\sigma}^2_t} \Bigg[\int_{\mathbf{x}_0} p(\mathbf{x}_0 | \mathbf{y}) p(\mathbf{x}_t | \mathbf{x}_0, \mathbf{y}) \bm{\mu}_t \mathbf{x}_0 d\mathbf{x}_0 \\
&\quad - \int_{\mathbf{x}_0} p(\mathbf{x}_0 | \mathbf{y}) p_t(\mathbf{x}_t | \mathbf{x}_0, \mathbf{y}) \mathbf{x}_t d\mathbf{x}_0\Bigg].
\end{align*}
Using Bayes' rule and recognizing the marginalization, we get:
\begin{align*}
\nabla_{\mathbf{x}_t} p(\mathbf{x}_t | \mathbf{y}) &= \frac{1}{\bm{\sigma}^2_t} \Bigg[\int_{\mathbf{x}_0} \bm{\mu}_t \mathbf{x}_0 p(\mathbf{x}_t | \mathbf{y}) p(\mathbf{x}_0 | \mathbf{x}_t, \mathbf{y})d\mathbf{x}_0 \\
&\quad - \mathbf{x}_t p(\mathbf{x}_t | \mathbf{y})\Bigg].
\end{align*}
\begin{align*}
\nabla_{\mathbf{x}_t} p(\mathbf{x}_t | \mathbf{y}) = \frac{1}{\bm{\sigma}^2_t} \left[\bm{\mu}_t p(\mathbf{x}_t | \mathbf{y})\E[\mathbf{x}_0 | \mathbf{x}_t, \mathbf{y}]-\mathbf{x}_t p(\mathbf{x}_t | \mathbf{y}) \right].
\end{align*}
\begin{align*}
\frac{\nabla_{\mathbf{x}_t} p(\mathbf{x}_t | \mathbf{y})}{p(\mathbf{x}_t | \mathbf{y})} = \frac{1}{\bm{\sigma}^2_t} \left[\bm{\mu}_t \E[\mathbf{x}_0 | \mathbf{x}_t, \mathbf{y}]- \mathbf{x}_t)\right].
\end{align*}
Using the identity property of logarithm  $\nabla_{\mathbf{x}} \log p(\mathbf{x}) = \nabla_{\mathbf{x}} p(\mathbf{x}) / p(\mathbf{x})$, we can rewrite:
\begin{align*}
\nabla_{\mathbf{x}_t} \log p(\mathbf{x}_t | \mathbf{y}) = \frac{1}{\bm{\sigma}^2_t} \left[ \bm{\mu}_t\E[\mathbf{x}_0 | \mathbf{x}_t, \mathbf{y}]-\mathbf{x}_t\right].
\end{align*}
From this, the posterior mean becomes:
\begin{align*} 
\E[\mathbf{x}_0 | \mathbf{x}_t, \mathbf{y}] =  \frac{\mathbf{x}_t +\bm{\sigma}^2_t \nabla_{\mathbf{x}_t} \log p(\mathbf{x}_t | \mathbf{y})}{\bm{\mu}_t}.
\end{align*}

This shows that the posterior mean of \( \mathbf{x}_0 \) conditioned on \( \mathbf{x}_t \) and \( \mathbf{y} \) now incorporates a scaling by \( \bm{\mu}_t \). By considering $\bm{\mu}_t = \sqrt{ \overline{\alpha}_t}$ and $\bm{\sigma}^2_t = 1-\overline{\alpha}_t$, we have then 
\begin{align*} 
\E[\mathbf{x}_0 | \mathbf{x}_t, \mathbf{y}] =  \frac{1}{\sqrt{ \overline{\alpha}_t}}({\mathbf{x}_t + (1- \overline{\alpha}_t) \nabla_{\mathbf{x}_t} \log p(\mathbf{x}_t | \mathbf{y})}).
\end{align*}

\subsubsection{Approximated Conditional Posterior Mean}\label{A42}
\begin{equation*}
\mathbb{E}[{\mathbf{x}_0|\mathbf{x}_t, \mathbf{y}}] = \frac{1}{{\sqrt{\overline{\alpha}_{t}}}}\big(\mathbf{x}_{t}+({1-\overline{\alpha}_{t}}){\nabla_{{{\mathbf{x}_t}}}\log p({\mathbf{x}_t}|\mathbf{y})}\big)    
\end{equation*} 
Considering \eqref{eq5} we have 
\begin{align*}
\mathbb{E}[{\mathbf{x}_0|\mathbf{x}_t, \mathbf{y}}] &= \frac{1}{{\sqrt{\overline{\alpha}_{t}}}}\big(\mathbf{x}_{t}+({1-\overline{\alpha}_{t}})(\nabla_{\mathbf{x}_t} \log p(\mathbf{x}_t) \\
&\quad + \nabla_{\mathbf{x}_t} \log p(\mathbf{y}|\mathbf{x}_t))\big)
\end{align*}
By knowing that $\textstyle{\nabla_{\mathbf{x}_t} \log p(\mathbf{x}_t) \simeq \frac{-1}{\sqrt{1-\overline{\alpha}_t}}\boldsymbol{\epsilon}_{\mathbf{\theta}}(\mathbf{x}_t, t)}$, then we get
\begin{align*}
\mathbb{E}[{\mathbf{x}_0|\mathbf{x}_t, \mathbf{y}}] &\simeq \frac{1}{{\sqrt{\overline{\alpha}_{t}}}}\bigg(\mathbf{x}_{t}+({1-\overline{\alpha}_{t}})\bigg( \frac{-1}{\sqrt{1-\overline{\alpha}_t}}\boldsymbol{\epsilon}_{\mathbf{\theta}}(\mathbf{x}_t, t) \\
&\quad + \nabla_{\mathbf{x}_t} \log p(\mathbf{y}|\mathbf{x}_t)\bigg)\bigg)
\end{align*}
which can be simplified further as 
\begin{align*}
\mathbb{E}[{\mathbf{x}_0|\mathbf{x}_t, \mathbf{y}}] &\simeq \frac{1}{{\sqrt{\overline{\alpha}_{t}}}}\big(\mathbf{x}_{t}- {\sqrt{1-\overline{\alpha}_t}}\boldsymbol{\epsilon}_{\mathbf{\theta}}(\mathbf{x}_t, t) \\
&\quad + (1-\overline{\alpha}_t)\nabla_{\mathbf{x}_t} \log p(\mathbf{y}|\mathbf{x}_t))\big)
\end{align*}
From approximation made by DPS \cite{chung2022diffusion}, that is, $\textstyle{\nabla_{\mathbf{x}_t} \log p(\mathbf{y}|\mathbf{x}_t) \simeq - \nicefrac{1}{\boldsymbol{\sigma}^2_{\mathbf{y}}} \nabla_{\mathbf{x}_t} \Vert{\mathbf{y}-\mathcal{A}(\mathbf{x}_{0|t}) }\Vert_2^2}$, we then get 
\begin{equation*}\label{eq13}
\tilde{\mathbf{x}}_{0|t}  \simeq \frac{1}{{\sqrt{\overline{\alpha}_{t}}}}  \Big[\mathbf{x}_t - \sqrt{1-\overline{\alpha}_{t}} \boldsymbol{\epsilon}_{\mathbf{\theta}}(\mathbf{x}_t, t) - 
{\color{orange}\zeta} \nabla_{\mathbf{x}_t}\Vert \mathbf{y} - \mathcal{A}\mathbf{x}_{0|t}\Vert_2^2
\Big].
\end{equation*}

\subsection{Comparing R-BGDM with Supervised Methods
}\label{A6}
Similar to other zero-shot inverse problem solvers \cite{chung2022diffusion, wang2022zero, kawar2022denoising}, R-BGDM is superior to existing supervised methods \cite{zhou2020dudornet, wei20202} in these dimensions:
\begin{itemize}
\item R-BGDM can be a zero-shot solver for diverse tasks, while supervised methods need to train separate models for each task and sampling patterns.

\item R-BGDM demonstrates robustness to patterns of undersampling and sparsification, whereas supervised techniques exhibit weak generalizability.

\item R-BGDM, akin to ScoreMed \cite{song2021solving} and Score-MRI \cite{chung2022score}, achieves notably enhanced results on medical datasets compared to supervised methods.

\end{itemize}
These claims are substantiated by the experimental results in Table \ref{tb5}. The results are reported from \cite{song2021solving, chung2022score}.
\begin{table*}[t!]
\scriptsize
\centering
\captionsetup{font=small, aboveskip=0pt, belowskip=0pt}
\caption{Comparison of R-BGDM against various supervised methods across multiple datasets.}
\begin{adjustbox}{max width=\textwidth}
\begin{tabular}{c|cc|cc|cc|cc|cc}
\hline
\multirow{2}{*}{Method} & \multicolumn{4}{c|}{BraTS-MRI} & \multicolumn{4}{c|}{fastMRI} & \multicolumn{2}{c}{LIDC-CT} \\ \cline{2-11} 
& \multicolumn{2}{c|}{8$\times$ ACR} & \multicolumn{2}{c|}{24$\times$ ACR} & \multicolumn{2}{c|}{4$\times$ ACR} & \multicolumn{2}{c|}{8$\times$ ACR} & \multicolumn{2}{c}{23 \textbf{Proj}} \\ \cline{2-11} 
& PSNR$\uparrow$         & SSIM$\uparrow$         & PSNR$\uparrow$         & SSIM$\uparrow$         & PSNR$\uparrow$         & SSIM$\uparrow$         & PSNR$\uparrow$         & SSIM$\uparrow$         & PSNR$\uparrow$         & SSIM$\uparrow$        \\ \hline
DuDoRNet \cite{zhou2020dudornet}                & 37.88            & 0.985        &18.46              & 0.662             &33.46             &0.856             &29.65              &0.777            &--              &--             \\
SIN-4c-PRN \cite{wei20202}               &--             &--            &--             &--           &--              &--             &--              &--             &30.48            &0.895 
\\
\hline 
\rowcolor{mylightblue}
\textbf{R-BGDM}           &\textbf{38.46}              &0.964              &\textbf{30.04}              &\textbf{0.887}              &\textbf{34.73}              &\textbf{0.875}              &\textbf{32.74}              &\textbf{0.835}              &\textbf{35.82}              &\textbf{0.911}             \\
\hline 
\end{tabular}
\end{adjustbox}
\label{tb5}
\end{table*}
\subsection{Additional Results}\label{A8}
\begin{figure*}[!ht]
\begin{adjustbox}{width=1\linewidth}
\begin{tikzpicture}
\draw (0, 0) node[inner sep=0] {\includegraphics[width=1\textwidth]{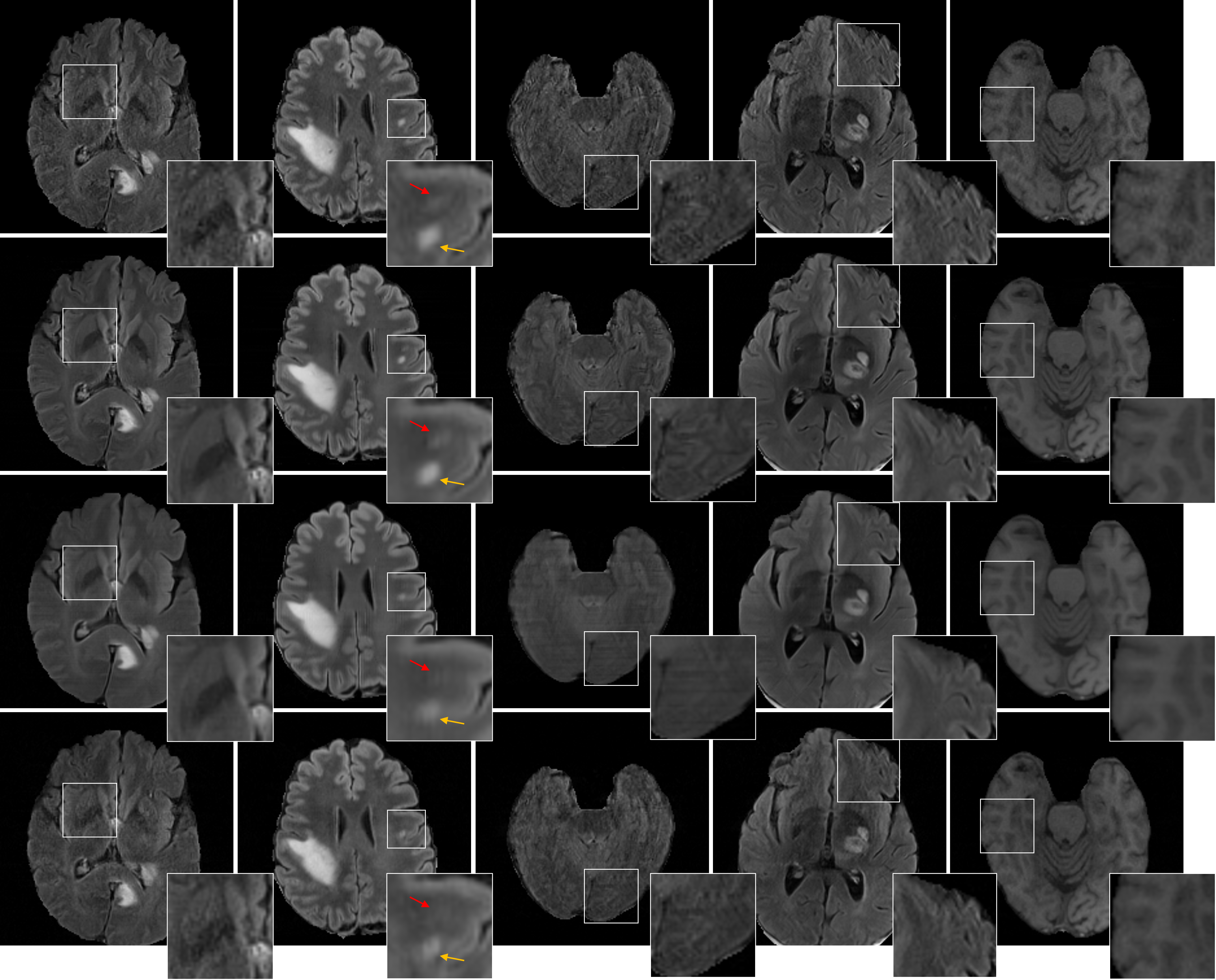}};
\draw (-9.8, 5.5) node{\fontsize{8pt}{9pt}\selectfont References};
\draw (-9.8, 2) node{\fontsize{8pt}{9pt}\selectfont R-BGDM};
\draw (-9.8, 1.7) node{\fontsize{8pt}{9pt}\selectfont {(\textbf{ours})}};
\draw (-9.8, -1.6) node{\fontsize{8pt}{9pt}\selectfont DDNM};
\draw (-9.8, -5.1) node{\fontsize{8pt}{9pt}\selectfont DPS};
\end{tikzpicture}
\end{adjustbox}
\caption{Additional results from undersampled MRI reconstruction on Brats at 8x acceleration rate.}
\label{fig9}
\end{figure*}
\begin{figure*}[!ht]
\begin{adjustbox}{width=1\linewidth}
\begin{tikzpicture}
\draw (0, 0) node[inner sep=0] {\includegraphics[width=0.94\textwidth]{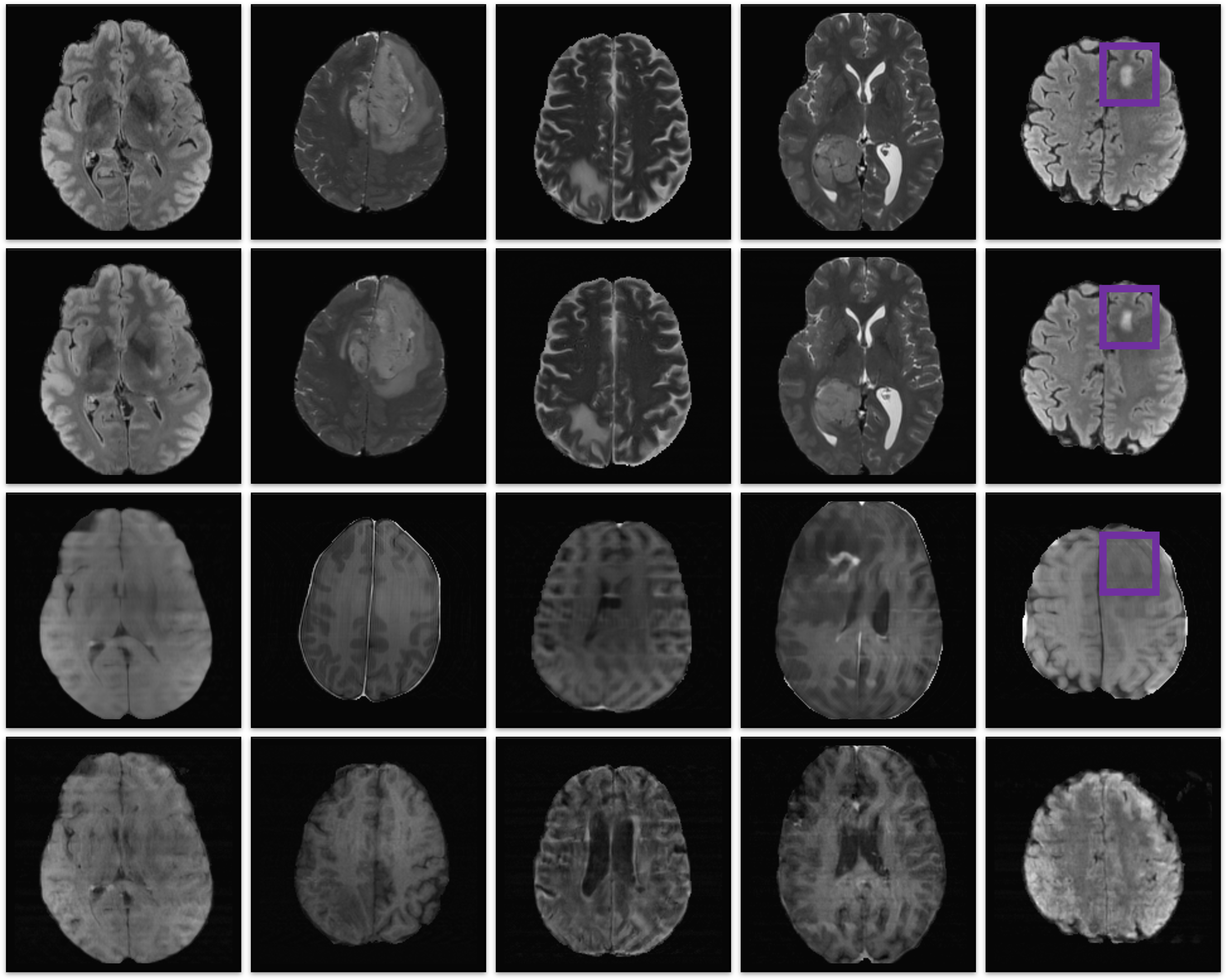}};
\draw (-9.1, 5.1) node{\fontsize{8pt}{9pt}\selectfont References};
\draw (-9.1, 1.5) node{\fontsize{8pt}{9pt}\selectfont R-BGDM};
\draw (-9.1, 1.2) node{\fontsize{8pt}{9pt}\selectfont {(\textbf{ours})}};
\draw (-9.1, -1.5) node{\fontsize{8pt}{9pt}\selectfont DDNM};
\draw (-9.1, -5.1) node{\fontsize{8pt}{9pt}\selectfont DPS};
\end{tikzpicture}
\end{adjustbox}
\caption{Additional results from undersampled MRI reconstruction on Brats at 24x acceleration rate.}
\label{fig10}
\end{figure*}
\begin{figure*}[!ht]
\begin{adjustbox}{width=1\linewidth}
\begin{tikzpicture}
\draw (0, 0) node[inner sep=0] {\includegraphics[width=0.94\textwidth]{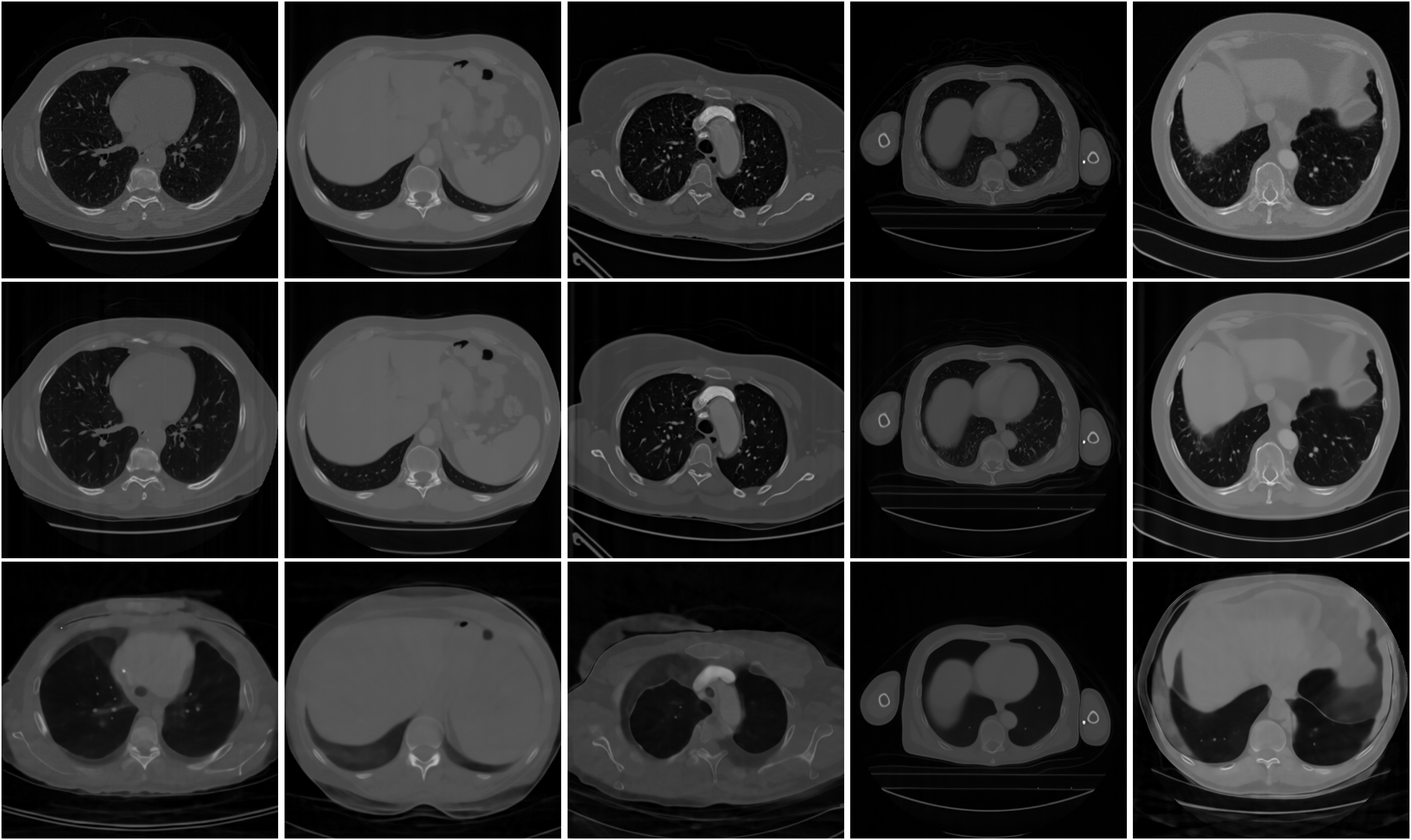}};
\draw (-9.2, 3.4) node{\fontsize{8pt}{9pt}\selectfont References};
\draw (-9.2, 0) node{\fontsize{8pt}{9pt}\selectfont R-BGDM};
\draw (-9.1, -0.3) node{\fontsize{8pt}{9pt}\selectfont {(\textbf{ours})}};
\draw (-9.2, -3.4) node{\fontsize{8pt}{9pt}\selectfont DDNM};
\end{tikzpicture}
\end{adjustbox}
\caption{Additional results from sparse-view CT reconstruction on LIDC dataset with 23 projections.}
\label{fig11}
\end{figure*}
\begin{figure*}[!ht]
\begin{adjustbox}{width=1\linewidth}
\begin{tikzpicture}
\draw (0, 0) node[inner sep=0] {\includegraphics[width=0.94\textwidth]{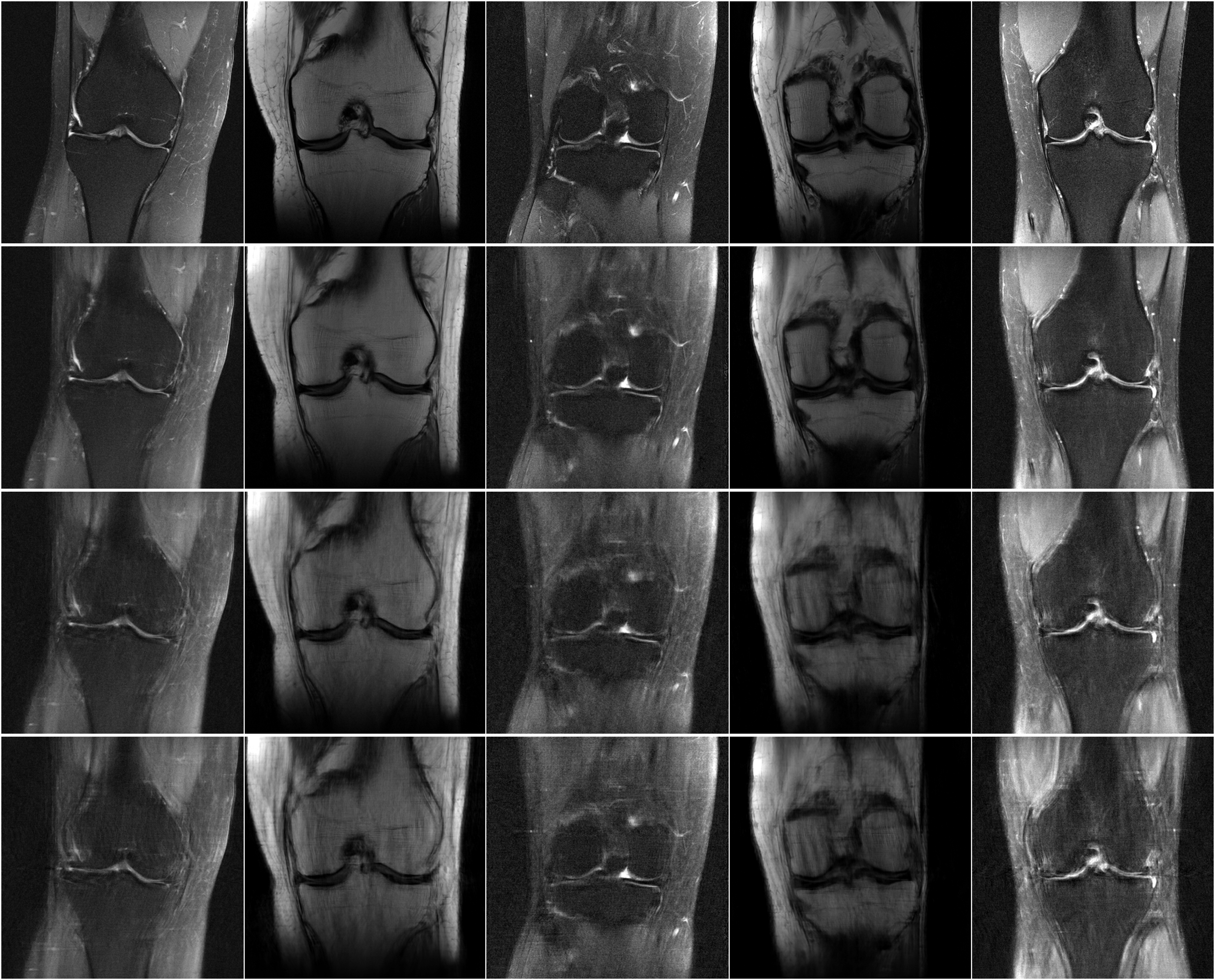}};
\draw (-9.2, 5.1) node{\fontsize{8pt}{9pt}\selectfont References};
\draw (-9.2, 1.4) node{\fontsize{8pt}{9pt}\selectfont BGDM};
\draw (-9.2, 1.1) node{\fontsize{8pt}{9pt}\selectfont {(\textbf{ours})}};
\draw (-9.2, -1.4) node{\fontsize{8pt}{9pt}\selectfont DDNM};
\draw (-9.2, -5.1) node{\fontsize{8pt}{9pt}\selectfont DPS};
\end{tikzpicture}
\end{adjustbox}
\caption{Additional reconstruction results for undersampled knee fastMRI at 4x acceleration rate.}
\label{fig12}
\end{figure*}
\begin{figure*}[!ht]
\begin{adjustbox}{width=1\linewidth}
\begin{tikzpicture}
\draw (0, 0) node[inner sep=0] {\includegraphics[width=0.94\textwidth]{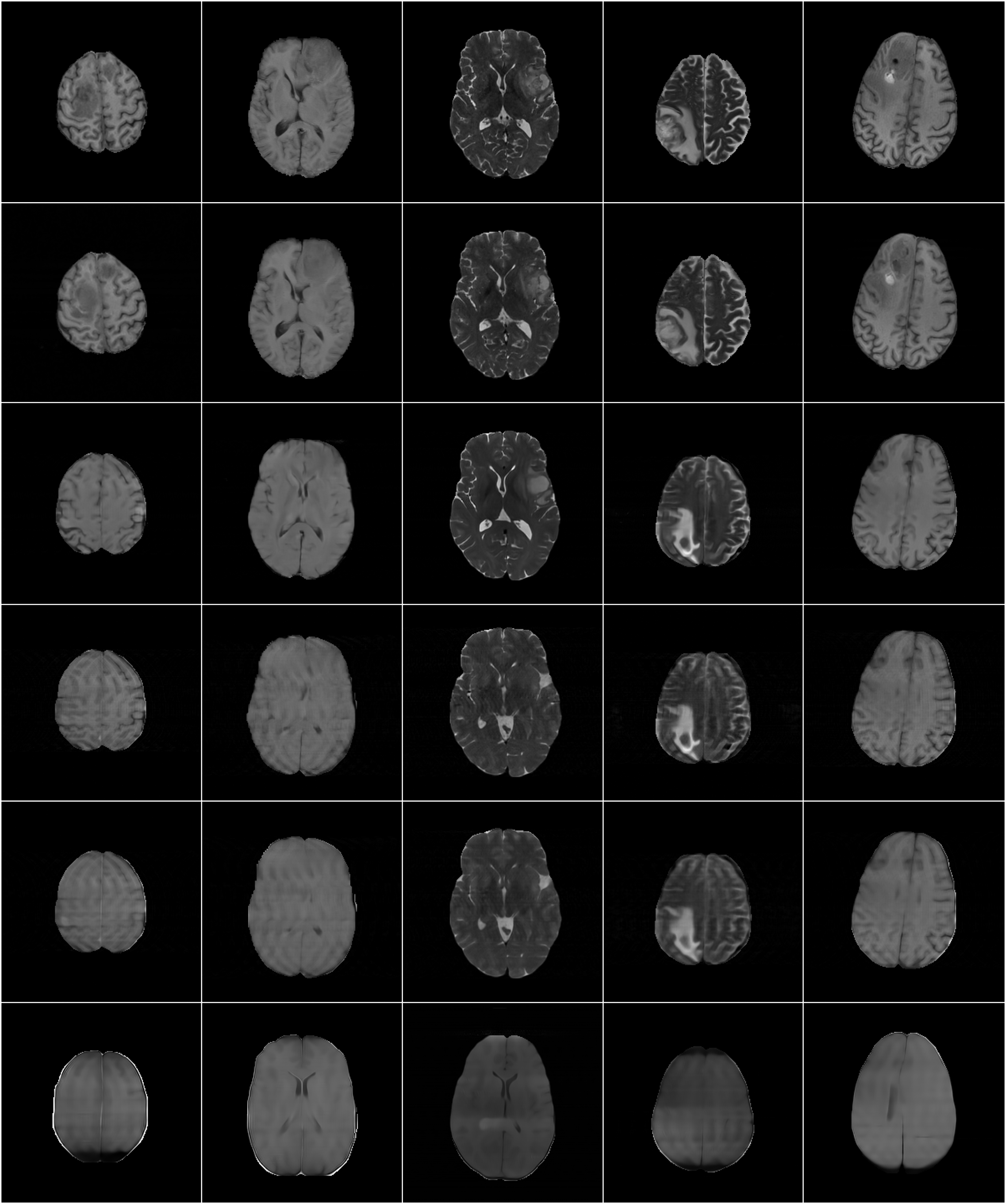}};
\draw (-9.2, 8.4) node{\fontsize{8pt}{9pt}\selectfont References};
\draw (-9.2, 5.1) node{\fontsize{8pt}{9pt}\selectfont R-BGDM};
\draw (-9.2, 1.6) node{\fontsize{8pt}{9pt}\selectfont BGDM};
\draw (-9.2, -1.6) node{\fontsize{8pt}{9pt}\selectfont {Ours\,$_{\text{no-\textbf{i}}}$}};
\draw (-9.2, -5.1) node{\fontsize{8pt}{9pt}\selectfont {Ours\,$_{\text{no-\textbf{ir}}}$}};
\draw (-9.2, -8.4) node{\fontsize{8pt}{9pt}\selectfont {Ours\,$_{\text{no-\textbf{or}}}$}};
\end{tikzpicture}
\end{adjustbox}
\caption{Additional results of our ablation study from undersampled MRI reconstruction on Brats at 24x acceleration rate.}
\label{13}
\end{figure*}
\begin{figure*}[!ht]
\begin{adjustbox}{width=1\linewidth}
\begin{tikzpicture}
\draw (0, 0) node[inner sep=0] {\includegraphics[width=0.94\textwidth]{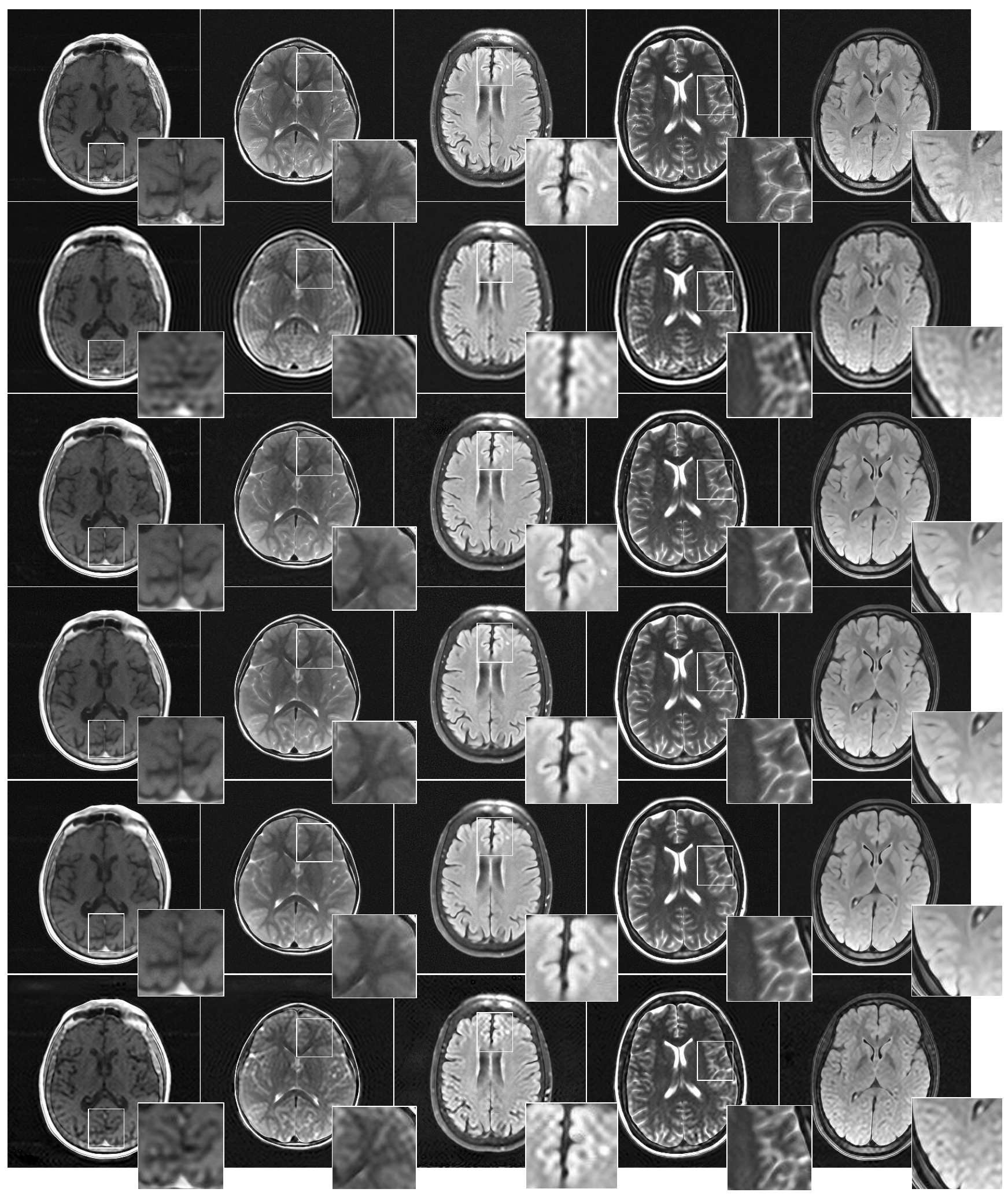}};
\draw (-9.5, 8.5) node{\fontsize{8pt}{9pt}\selectfont References};
\draw (-9.5, 8.2) node{\fontsize{8pt}{9pt}\selectfont (high-resolution)};
\draw (-9.5, 5.2) node{\fontsize{8pt}{9pt}\selectfont Low-resolution};
\draw (-9.5, 1.8) node{\fontsize{8pt}{9pt}\selectfont R-BGDM};
\draw (-9.5, 1.5) node{\fontsize{8pt}{9pt}\selectfont (\textbf{ours})};
\draw (-9.5, -1.3) node{\fontsize{8pt}{9pt}\selectfont BGDM};
\draw (-9.5, -1.6) node{\fontsize{8pt}{9pt}\selectfont (\textbf{ours})};
\draw (-9.5, -4.8) node{\fontsize{8pt}{9pt}\selectfont DDNM};
\draw (-9.5, -8.1) node{\fontsize{8pt}{9pt}\selectfont DPS};
\end{tikzpicture}
\end{adjustbox}
\caption{Additional results for super-resolution fastMRI at 16x acceleration rate.}
\label{14}
\end{figure*}

\end{appendices}

\end{document}